\documentclass[twoside,11pt]{article}

\usepackage{float}

\usepackage[utf8]{inputenc}
\usepackage{graphicx, subfigure}
\usepackage[top=1in,bottom=1in,left=1in,right=1in,footnotesep=0.5in]{geometry}
\usepackage[sort&compress]{natbib} 
\usepackage{amsfonts, amsmath, amssymb, amsthm, constants, bbm}
\usepackage{MnSymbol}
\usepackage{booktabs} 
\usepackage{enumitem}

\usepackage{dsfont} 

\usepackage[utf8]{inputenc}
\usepackage{MnSymbol}
\usepackage{booktabs} 
\usepackage{enumitem}
\usepackage{amsfonts, amsmath, amsthm}

\usepackage{color}
\definecolor{darkred}{RGB}{150,0,0}
\definecolor{darkgreen}{RGB}{0,150,0}
\definecolor{darkblue}{RGB}{0,0,150}

\usepackage{hyperref}
\hypersetup{colorlinks=true, linkcolor=darkred, citecolor=darkgreen, urlcolor=darkblue}
\usepackage{url}

\usepackage[sort&compress]{natbib}

\def\d{{\rm d}}

\newtheorem{thm}{Theorem}

\theoremstyle{remark}

\newtheorem{rem}{Remark}

\def\beq{\begin{equation}} 
\def\eeq{\end{equation}}
\def\beqn{\begin{eqnarray*}}
\def\eeqn{\end{eqnarray*}}
\def\Bitem{\begin{itemize}\setlength{\itemsep}{.2in}}
\def\bitem{\begin{itemize}\setlength{\itemsep}{.05in}}
\def\eitem{\end{itemize}}
\def\Benum{\begin{enumerate}\setlength{\itemsep}{.2in}}
\def\benum{\begin{enumerate}\setlength{\itemsep}{.05in}}
\def\eenum{\end{enumerate}}
\def\bmult{\begin{multline*}}
\def\emult{\end{multline*}}
\def\bcenter{\begin{center}}
\def\ecenter{\end{center}}
\def\bframe{\begin{frame}}
\def\eframe{\end{frame}}

\newcommand{\thmref}[1]{Theorem~\ref{thm:#1}}

\newcommand{\secref}[1]{Section~\ref{sec:#1}}
\newcommand{\figref}[1]{Figure~\ref{fig:#1}}
\newcommand{\tabref}[1]{Table~\ref{tab:#1}}

\DeclareMathOperator*{\argmax}{arg\, max}

\DeclareMathOperator*{\esssup}{ess\, sup}
\DeclareMathOperator*{\essinf}{ess\, inf}

\def\cC{\mathcal{C}}
\def\cD{\mathcal{D}}
\def\cE{\mathcal{E}}

\def\cG{\mathcal{G}}

\def\cM{\mathcal{M}}

\def\cS{\mathcal{S}}

\def\bv{\mathbf{v}}

\def\bbR{\mathbb{R}}

\newcommand{\E}{\operatorname{\mathbb{E}}}
\renewcommand{\P}{\operatorname{\mathbb{P}}}

\def\eps{\varepsilon}

\def\comp{\mathsf{c}}
\def\1{\mathbbm{1}}
\newcommand{\IND}[1]{\mathds{1}\{ #1 \}}

\begin{document}

\thispagestyle{empty}
\title{Extending the Patra--Sen Approach to Estimating the Background Component in a Two-Component Mixture Model}
\author{
Ery Arias-Castro\footnote{Department of Mathematics, University of California, San Diego, USA \newline \indent \quad \url{https://math.ucsd.edu/\~eariasca}} 
\and He Jiang\footnote{Department of Mathematics, University of California, San Diego, USA \newline \indent \quad \url{https://math.ucsd.edu/people/graduate-students/}}
}
\date{}
\maketitle

\begin{abstract}
\cite{patra} consider a two-component mixture model, where one component plays the role of background while the other plays the role of signal, and propose to estimate the background component by simply `maximizing' its weight. While in their work the background component is a completely known distribution, we extend their approach here to three emblematic settings: when the background distribution is symmetric; when it is monotonic; and when it is log-concave. In each setting, we derive estimators for the background component, establish consistency, and provide a confidence band. While the estimation of a background component is straightforward when it is taken to be symmetric or monotonic, when it is log-concave its estimation requires the computation of a largest concave minorant, which we implement using sequential quadratic programming.
Compared to existing methods, our method has the advantage of requiring much less prior knowledge on the background component, and is thus less prone to model misspecification. We illustrate this methodology on a number of synthetic and real datasets.
\end{abstract}

\section{Introduction}  \label{sec:intro}

\subsection{Two component mixture models}

Among mixture models, two-component models play a special role. In robust statistics, they are used to model contamination, with the main component representing the inlier distribution, while the remaining component representing the outlier distribution \citep{hettmansperger2010robust, huber2009robust, tukey1960survey, huber1964robust}. In that kind of setting, the contamination is a nuisance and the goal is to study how it impacts certain methods for estimation or testing, and also to design alternative methods that behave comparatively better in the presence of contamination.

In multiple testing, the background distribution plays the role of the distribution assumed (in a simplified framework) to be common to all test statistics under their respective null hypotheses, while the remaining component plays the role of the distribution assumed of the test statistics under their respective alternative hypotheses \citep{efron2001empirical, genovese2002operating}.
In an ideal situation where the $p$-values can be computed exactly and are uniformly distributed on $[0,1]$ under their respective null hypotheses, the background distribution is the uniform distribution on $[0,1]$. Compared to the contamination perspective, here the situation is in a sense reverse, as we are keenly interested in the component other than the background component. 
We adopt this multiple testing perspective in the present work. 

\subsection{The Patra--Sen approach}

Working within the multiple testing framework, \cite{patra} posed the problem of estimating the background component as follows. They operated under the assumption that the background distribution is completely known --- a natural choice in many pracical situations, see for example the first two situations in \secref{symmetric_read_data}.
Given a density $f$ representing the density of all the test statistics combined, and letting $g_0$ denote a completely known density, define
\begin{equation}
\label{theta_0_definition}
\theta_0 := \sup \{t: f \ge t g_0\}.
\end{equation}
Note that $\theta_0 \in [0,1]$.
Under some mild assumptions on $f$, the supremum is attained, so that $f$ can be expressed as the following two-component mixture:
\begin{equation}
f = \theta_0 g_0 + (1-\theta_0) u,
\end{equation}
for some density $u$.
\cite{patra} aim at estimating $\theta_0$ defined in \eqref{theta_0_definition} based on a sample from the density $f$, and implement a slightly modified plug-in approach.
Even in this relatively simple setting where the background density --- the completely known density $g_0$ above --- is given, information on $\theta_0$ can help improve inference in a multiple testing situation as shown early on by \cite{storey2002direct}, and even earlier by \cite{benjamini2000adaptive}.

\subsection{Our contribution}
We find the Patra--Sen approach elegant, and in the present work extend it to settings where the background distribution (also referred to as the null distribution) --- not just the background proportion --- is unknown. For an approach that has the potential to be broadly applicable, we consider three emblematic settings where the background distribution is in turn assumed to be symmetric (\secref{symmetric}), monotone (\secref{monotone}), or log-concave (\secref{log-concave}). 
Each time, we describe the estimator for the background component (proportion and density) that the Patra-Sen approach leads to, and study its consistency and numerical implementation. We also provide a confidence interval for the background proportion and a simultaneous confidence band for the background density. In addition, in the log-concave setting, we provide a way of computing the largest concave minorant. We address the situation where the background is specified incorrectly, and mention other extensions, including combinations of these settings and in multivariate settings, in \secref{discussion}.


\subsection{More related work in multiple testing}
The work of \cite{patra} adds to a larger effort to estimate the proportion and/or the density of the null component in a multiple testing scenario. This effort dates back, at least, to early work on false discovery control \citep{benjamini1995controlling} where (over-)estimating the proportion of null hypotheses is crucial to controlling the FDR and related quantities \citep{storey2002direct, benjamini2000adaptive, genovese2004stochastic}.

Directly focusing on the estimation of the null proportion, \cite{langaas2005estimating} consider a setting where the $p$-values are uniform in $[0,1]$ under their null hypotheses and have a monotone decreasing density under their alternative hypotheses, while \cite{meinshausen2006estimating} do not assume anything of the alternative distribution and propose an estimator which is similar in spirit to that of \cite{patra}.
\cite{jin2007estimating} and \cite{jin2008proportion} consider a Gaussian mixture model and approach the problem via the characteristic function --- a common approach in deconvolution problems. 
Gaussian mixtures are also considered in \citep{efron2007size, efron2012large, cai2010optimal}, where the Gaussian component corresponding to the null has unknown parameters that need to be estimated. 

Some references to estimating the null component that have been or could be applied in the context of multiple testing are given in \cite[Ch 5]{efron2012large}.
Otherwise, we are also aware of the very recent work of \cite{roquain2020}, where in addition to studying the `cost' of having to estimate the parameters of the null distribution when assumed Gaussian, also consider the situation where null distribution belongs to a given location family, and further, propose to estimate the null distribution under an upper bound constraint on the proportion of non-nulls in the mixture model.

\begin{rem}
Much more broadly, all this connects with the vast literature on Gaussian mixture models \citep{cohen1967estimation, lindsay1993multivariate} and on mixture models in general \citep{mclachlan2004finite, mclachlan1988mixture, mclachlan2019finite, lindsay1995mixture}, including two-component models \citep{shen2018mm, bordes2006semiparametric, ma2015flexible, gadat2020parameter}. 
\end{rem}

\section{Symmetric background component}
\label{sec:symmetric}

We start with what is perhaps the most natural nonparametric class of null distributions: the class of symmetric distributions about the origin. Unlike \cite{roquain2020}, who assume that the null distribution is symmetric around an unknown location that needs to be estimated but is otherwise known, i.e., its `shape' is known, we assume that the shape is unknown. We do assume that the center of symmetry is known, but this is for simplicity, as an extension to an unknown center of symmetry is straightforward (see our numerical experiments in \secref{symmetric_experiments}). 
Mixtures of symmetric distributions are considered in \citep{hunter2007inference}, but otherwise, we are not aware of works estimating the null distribution under an assumption of symmetry in the context of multiple testing. 
For works in multiple testing that assume that the null distribution is symmetric but unknown, but where the goal is either testing the global null hypothesis or controlling the false discovery rate, see \citep{arias2017distribution, arias2017distribution_fdr}.

Following the footsteps of \cite{patra}, we make sense of the problem by defining for a density $f$ the following:
\begin{equation}
\label{symmetric_pi0_definition}
\pi_0 := \sup \big\{\pi: \exists g \in \cS \text{ s.t. } f - \pi g \geq 0 \text{ a.e.}\big\},
\end{equation}
where $\cS$ is the class of even densities (i.e., representing a distribution that is symmetric about the origin).
Note that $\pi_0 \in [0,1]$ is well-defined for any density $f$, with $\pi_0 = 1$ if and only if $f$ itself is symmetric.

\begin{thm}
We have
\begin{align}
\label{symmetric_pi0_value} 
\pi_0 = \int_{-\infty}^\infty h_0(x) d x,
&& h_0(x) := \min \{ f(x), f(-x) \}.
\end{align}
Moreover, if $\pi_0 > 0$ the supremum in \eqref{symmetric_pi0_definition} is attained by the following density and no other\,\footnote{~As usual, densities are understood up to sets of zero Lebesgue measure.} :
\begin{equation}
\label{symmetric_g0} 
g_0(x) := \frac{h_0(x)}{\pi_0}.
\end{equation}
\end{thm}

\begin{proof}
The parameter $\pi_0$ can be equivalently defined as 
\begin{align}
\pi_0 = \sup \big\{\textstyle\int h: \text{$h$ is even and $0 \le h \le f$ a.e.}\big\}.
\end{align}
Note that $h_0$, as defined in the statement, satisfies the above conditions, implying that $\pi_0 \ge \int h_0$. 
Take $h$ satisfying these same conditions, namely, $h(x) = h(-x)$ and $0 \le h(x) \le f(x)$ for almost all $x$. Then, for almost any $x$, $h(x) \le f(x)$ and $h(-x) \le f(-x)$, implying that $h(x) \le f(x) \wedge f(-x) = h_0(x)$. 
(Here and elsewhere, $a \wedge b$ is another way of denoting $\min(a, b)$.) 
Hence, $\int h \le \int h_0$ with equality if and only if $h = h_0$ a.e., in particular implying that $\pi_0 \le \int h_0$. We have thus established that $\pi_0 = \int h_0$, and also that $\int h = \pi_0$ if and only if $h = h_0$ a.e.. This not only proves \eqref{symmetric_pi0_value}, but also \eqref{symmetric_g0}, essentially by definition.  
\end{proof}

We have thus established that, in the setting of this section, the background component as defined above is given by 
\begin{equation}
h_0(x) = \pi_0 g_0(x) = \min \{ f(x), f(-x) \},
\end{equation}
and $f$ can be expressed as a mixture of the background density and another, unspecified, density $u$, as follows:
\begin{equation}
f = \pi_0 g_0 + (1-\pi_0) u.
\end{equation}
The procedure is summarized in \tabref{symmetric_algorithm}. 
An illustration of this decomposition is shown in \figref{illustration_symmetric}. By construction, the density $u$ is such that it has no symmetric background component in that, for almost every $x$, $u(x) = 0$ or $u(-x) = 0$.

\begin{table}[htpb]
\centering
\caption{Symmetric background computation.}
\label{tab:symmetric_algorithm}
\bigskip
\setlength{\tabcolsep}{0.22in}
\begin{tabular}{ p{0.9\textwidth}  }
\toprule

{\textbf{inputs}: density ${f}$, given center of symmetry $c_0$ or candidate center points $\{c_1, c_2, \dots, c_k$\}} \\ \midrule

\textbf{if} center of symmetry is not provided \textbf{then}

\hspace{3mm} \textbf{for} $i=1,\dots,k$ \textbf{do}

\hspace{3mm} \hspace{3mm} $h_i(x) = \min \{ {f}(x), {f}(2c_i - x) \}$

\hspace{3mm} \hspace{3mm} $\pi_i(x) = \int_{-\infty}^{\infty} h_i(x) dx $

\hspace{3mm} $\beta = \argmax_i \pi_i(x)$

\hspace{3mm} $c_0 = c_{\beta}$

$h_0(x) = \min \{ {f}(x), {f}(2c_0 - x) \}$

$\pi_0 = \int_{-\infty}^{\infty} h_0(x) dx$

$g_0(x) = h_0(x) / \pi_0$ \\

\midrule
\textbf{return} $c_0, \pi_0, g_0, h_0$
\\

\bottomrule
\end{tabular}
\end{table}

\begin{figure}[htpb]
    \centering
     
    \centering
    
    \subfigure[Decomposition of $f$ with center of symmetry specified as $0$ (dotted).]{\label{fig:a}\includegraphics[scale=0.35]{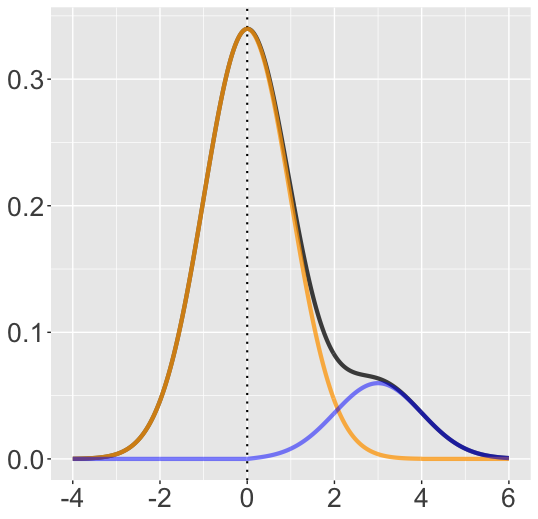}}\qquad
    \centering
    \subfigure[Decomposition of $f$ with center of symmetry, $0.04$ (dotted), found by maximization. ]{\label{fig:b}\includegraphics[scale=0.35]{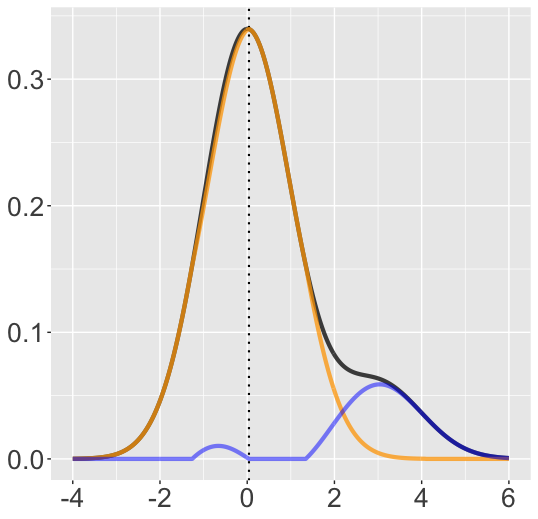}}
    \caption{The density $f$ of the Gaussian mixture $0.85 \text{ } \mathcal{N}(0, 1) + 0.15 \text{ } \mathcal{N} (0, 1)$, in black, and its decomposition into $\pi_0 g_0$, in orange, and $(1 - \pi_0) u$, in blue. We specify the center of symmetry as $0$ on the left, and we do not specify the center of symmetry on the right. Notice $\pi_0 = 0.850$ on the left and $\pi_0 = 0.860$ on the right. }
    \label{fig:illustration_symmetric}
\end{figure}

\subsection{Estimation and consistency}

When all we have to work with is a sample --- $x_1, x_2, \dots, x_n \in \bbR$ --- we adopt a straightforward plug-in approach: We estimate the density $f$, obtaining $\hat f$, and apply the procedure of \tabref{symmetric_algorithm}, meaning, we compute $\hat h_0(x) := \min \{ \hat f(x), \hat f(-x) \}$. If we want estimates for the background density and proportion, we simply return $\hat\pi_0 := \int \hat h_0$ and $\hat g_0 := \hat h_0/\hat\pi_0$. (By convention, we set $\hat g_0$ to the standard normal distribution if $\hat\pi_0 = 0$.) 

We say that $\hat f = \hat f_n$ is locally uniformly consistent for $f$ if $\E[\esssup_{x \in I} |\hat f_n(x) - f(x)|] \to 0$ as $n\to\infty$ for any bounded interval $I$.
(Here and elsewhere, $\esssup_{x \in I} f(x)$ denotes the essential supremum of $f$ over the set $I$.)
We note that this consistency condition is satisfied, for example, when $f$ is continuous and $\hat f$ is the kernel density estimator with the Gaussian kernel and bandwidth chosen by cross-validation \citep{chow1983consistent}. 

\begin{thm}
\label{thm:symmetric_consistency}
Suppose that $\hat f$ is a true density and locally uniformly consistent for $f$. Then $\hat h_0$ is locally uniformly consistent for $h_0$ and $\hat \pi_0$ is consistent, and if $\pi_0 > 0$, then $\hat g_0$ is locally uniformly consistent for $g_0$.
\end{thm}

\begin{proof}
All the limits that follows are as the sample size diverges to infinity.
We rely on the elementary fact that, for $a_1, a_2, b_1, b_2 \in \bbR$,
\begin{equation}
\big|\min\{a_1, b_1\} - \min\{a_2, b_2\}\big|
\le \max\{|a_1 - a_2|, |b_1 - b_2|\},
\end{equation}
to get that, for all $x$,
\begin{equation}
|\hat h_0(x) - h_0(x)|
\le \max\{|\hat f(x) - f(x)|, |\hat f(-x) - f(-x)|\},
\end{equation}
implying that $\hat h_0$ is locally uniformly consistent for $h_0$.
To be sure, take a bounded interval $I$, which we assume to be symmetric without loss of generality. Then 
\begin{align*}
\esssup_{x \in I} |\hat h_0(x) - h_0(x)|
&\le \esssup_{x \in I} |\hat f(x) - f(x)| \vee \esssup_{x \in I} |\hat f(-x) - f(-x)| \\
&= \esssup_{x \in I} |\hat f(x) - f(x)|,
\end{align*}
and we then use the fact that $\E[\esssup_I |\hat f - f|] \to 0$.
(Here and elsewhere, $a \vee b$ is another way of denoting $\max(a, b)$.)


To conclude, it suffices to show that $\hat\pi_0$ is consistent for $\pi_0$.
Fix $\eps > 0$ arbitrarily small. There is a bounded interval $I$ such that $\int_I f \ge 1 - \eps$. 
Then, by the fact that $0\le h_0 \le f$ a.e., 
\[
\int_I h_0 \le \pi_0 = \int h_0 \le \int_I h_0 + \int_{I^\comp} f \le \int_I h_0 + \eps.
\]
($I^\comp$ denotes the complement of $I$, meaning, $I^\comp = \bbR \setminus I$.)
Similarly, by the fact that $0\le \hat h_0 \le \hat f$ a.e., 
\[
\int_I \hat h_0 \le \hat\pi_0 = \int \hat h_0 \le \int_I \hat h_0 + \int_{I^\comp} \hat f.
\]
From this we gather that
\begin{equation*}
\int_I \hat h_0 - \int_I h_0 - \eps
\le
\hat\pi_0 - \pi_0
\le \int_I \hat h_0 - \int_I h_0 + \int_{I^\comp} \hat f.
\end{equation*}
Thus consistency of $\hat\pi_0$ follows if we establish that $\limsup \int_{I^\comp} \hat f \le \eps$ and that $\int_I \hat h_0 - \int_I h_0 \to 0$.
The former comes from the fact that 
\begin{align*}
\int_{I^\comp} \hat f 
&= \int_{I^\comp} \hat f - \int_{I^\comp} f + \int_{I^\comp} f \\
&\le \int_{I} (\hat f - f) + \eps \\
&\le |I| \esssup_I |\hat f - f| + \eps \to \eps,
\end{align*}
using the fact that $f$ and $\hat f$ are densities and that $\int_{I^\comp} f \le \eps$.
($|I|$ denotes the Lebesgue measure of $I$, meaning its length when $I$ is an interval.)
For the latter, we have
\begin{align*}
\left|\int_I \hat h_0 - \int_I h_0\right|
&\le \int_I |\hat h_0 - h_0|
&\le |I| \esssup_I |\hat h_0 - h_0| \to 0,
\end{align*}
having already established that $\hat h_0$ is locally uniformly consistent for $h_0$.
\end{proof}

\paragraph{Confidence interval and confidence band}
Beyond mere pointwise consistency, suppose that we have available a confidence band for $f$, which can be derived under some conditions on $f$ from a kernel density estimator --- see \citep{chen2017tutorial} or \citep[Ch 6.4]{gine2021mathematical}.

\begin{thm}
Suppose that for some $\alpha \in (0,1)$, we have at our disposal $\hat f_l$ and $\hat f_u$ such that 
\begin{equation}
\label{conf1}
\P \big(\hat f_l(x) \leq f(x) \leq \hat f_u(x), \text{ for almost all $x$}\big) \geq 1 - \alpha.
\end{equation}
Then, with probability at least $1-\alpha$, 
\begin{align}
\label{conf2}
\hat\pi_l \le \pi_0 \le \hat\pi_u, 
&& \hat g_l \le g_0 \le \hat g_u \text{ a.e.},
\end{align}
where 
\begin{align*}
\hat\pi_l := \int \hat h_l,
&& \hat\pi_u := \int \hat h_u, 
&& \hat h_l(x) := \min \{ \hat{f}_l(x), \hat{f}_l(-x)\},
&& \hat h_u(x) := \min \{ \hat{f}_u(x), \hat{f}_u(-x)\}.
\end{align*}

\end{thm}

\begin{proof}
Let $\Omega$ be the event that $\hat f_l(x) \leq f(x) \leq \hat f_u(x)$ for almost all $x$, and note that $\P(\Omega) \ge 1-\alpha$ by assumption. Assuming that $\Omega$ holds, we have, for almost all $x$,
\begin{align*}
\hat{f}_l (x) \leq f(x) \leq \hat{f}_u (x),
&& \hat{f}_l (-x) \leq f(-x) \leq \hat{f}_u (-x),
\end{align*}
and taking the minimum in the corresponding places yields
\begin{equation}
\label{symmetric_minimum_bound}
\hat h_l(x) \leq h_0(x) \leq \hat h_u(x).
\end{equation}
Everything else follows immediately from this.
\end{proof}

In words, we apply the procedure of \tabref{symmetric_algorithm} to the lower and upper bounds, $\hat f_l$ and $\hat f_u$. If the center of symmetry is not provided, then we determine it based on $\hat f$, and then use that same center for $\hat f_l$ and $\hat f_u$.

\subsection{Numerical experiments}
\label{sec:symmetric_experiments}

In this subsection we provide simulated examples when the background distribution is taken to be symmetric. We acquire $\hat{f}$ by kernel density estimation using the Gaussian kernel with bandwidth selected by cross-validation \citep{rudemo1982empirical, stone1984asymptotically, arlot2010survey, silverman1986density, sheather1991reliable}, where the cross-validated bandwidth selection is implemented in the \textsf{kedd} package \citep{guidoum2015kernel}. The consistency of this density estimator has been proven in \citep{chow1983consistent}. Although there are many methods that provide confidence bands for kernel density estimator, for example \citep{bickel1973some, gine2010confidence}, for consideration of simplicity and intuitiveness, our simultaneous confidence band (in the form of $\hat{f}_l$ and $\hat{f}_u$) used in the experiments is acquired from bootstrapping a debiased estimator of the density as proposed in  \citep{cheng2019nonparametric}. For a comprehensive review on the area of kernel density and confidence bands, we point the reader to the recent survey paper \citep{chen2017tutorial} and textbook \citep[Ch 6.4]{gine2021mathematical}.

In the experiments below, we carry out method and report the estimated proportion $\hat{\pi}_0$, as well as its $95\%$ confidence interval $(\hat{\pi}_0^{\rm{L}}, \hat{\pi}_0^{\rm{U}})$. 
We consider both the situation where the center of symmetry is given, and the situation where it is not. 
In the latter situation, we also report the background component's estimated center, which is selected among several candidate centers and chosen as the one giving the largest symmetric background proportion. 

We are not aware of other methods for estimating the quantity $\pi_0$, but we provide a comparison with several well-known methods that estimate similar quantities. 
To begin with, we consider the method of \cite{patra}, which estimates the quantity $\theta_0$ as defined in $\eqref{theta_0_definition}$.
We let $\hat{\theta}_0^{\rm{PSC}}$, $\hat{\theta}_0^{\rm{PSH}}$, $\hat{\theta}_0^{\rm{PSB}}$ denote the constant, heuristic, and $95\%$ upper bound estimator on $\theta_0$, respectively. 
We then consider estimators of the quantity $\theta$, the actual proportion of the given background component $f_b$ in the mixture density
\begin{equation}
\label{theta_definition}
f = \theta f_b + (1-\theta) u,
\end{equation}
where $u$ denotes the unknown component. 
Note that $\theta_0$ and $\pi_0$ may be different from $\theta$.
This mixture model may not be identifiable in general, and we discuss this issue down below. 
We also consider the estimator of \citep{efron2007size}, denoted $\hat{\theta}^{\rm{E}}$, and implemented in package \textsf{locfdr}\,\footnote{~\url{https://cran.r-project.org/web/packages/locfdr/index.html}}. This method requires the unknown component to be located away from $0$, and to be have heavier tails than the background component. 
In addition, when the $p$-values are known, \cite{meinshausen2006estimating} provide a $95\%$ upper bound on the proportion of the null component, and we include that estimator also, denoted $\hat{\theta}^{\rm{MR}}$, and implemented in package \textsf{howmany}\,\footnote{~\url{https://cran.r-project.org/web/packages/howmany/howmany.pdf}}. Finally, when the distribution is assumed to be a Gaussian mixture, \cite{cai2010optimal} provide an estimator, denoted $\hat{\theta}^{\rm{CJ}}$, when the unknown component is assumed to have larger standard deviation than the background component. $\hat{\theta}^{\rm{CJ}}$ requires the specification of a parameter $\gamma$, and following the advice given by the authors, we select $\gamma = 0.2$. 
Importantly, unlike our method, these other methods assume knowledge of the background distribution. (Note that the methods of \cite{efron2007size} and \cite{cai2010optimal} do not necessitate full knowledge of the background distribution, but we provide them with that knowledge in all the simulated datasets.)
We summarize the methods used in our experiments in \tabref{other_methods_summary}.
For experiments in situations where the background component is misspecified, we invite the reader to \secref{incorrect_background_subsection}.

\begin{table}[htpb]
\centering\small
\caption{Summary of the methods considered in our experiments.}
\label{tab:other_methods_summary}
\bigskip
\setlength{\tabcolsep}{0.12in}
\begin{tabular}{p{0.14\textwidth} 
p{0.2\textwidth} 
p{0.25\textwidth} 
p{0.26\textwidth}  
}
\toprule
{\bf Estimator} & {\bf Reference} & {\bf Description } & {\bf Background information needed} \\ 
\midrule
$\hat{\pi}_0, \hat{\pi}_0^{\rm{L}}, \hat{\pi}_0^{\rm{U}}$ & Current paper & Estimator of $\pi_0$ with $95\%$ lower and upper confidence bounds. & Requires background distribution to be symmetric (note the requirement becomes monotonic in \secref{monotone} or log-concave in \secref{log-concave}). \\ 
\midrule
$\hat{\theta}_0^{\rm{PSC}}, \hat{\theta}_0^{\rm{PSH}}, \hat{\theta}_0^{\rm{PSB}}$ & \citep{patra} & Constant, heuristic, and $95\%$ upper bound estimates of  $\theta_0$. & Requires complete knowledge on the background distribution.\\ 
\midrule
$\hat{\theta}_0^{\rm{E}}$ & \citep{efron2007size} & Estimator of $\theta$. & Either requires full knowledge of background distribution or can estimate the background distribution when it has shape similar to a Gaussian distribution centered around $0$. \\ 
\midrule
$\hat{\theta}_0^{\rm{MR}}$ & \citep{meinshausen2006estimating} & $95\%$ upper bound of $\theta$.  & Requires complete knowledge on the background distribution. \\ 
\midrule
$\hat{\theta}_0^{\rm{CJ}}$ & \citep{cai2010optimal} & Estimator of $\theta$. & Either requires full knowledge of background distribution or can estimate the background distribution when it is Gaussian. \\

\bottomrule
\end{tabular}
\end{table}

We consider four different situations as listed in \tabref{symmetric_simulation_situations}. Each situation's corresponding $\theta$, $\theta_0$, and $\pi_0$, defined as in \eqref{symmetric_pi0_definition}, are also presented, where $\theta_0$ and $\pi_0$ are obtained numerically based on knowledge of $f$. For each model, we generate a sample of size $n = 1000$ and compute all the estimators described above. We repeat this process $1000$ times. 
We transform the data accordingly when applying the comparison methods. 
The result of our experiment are reported, in terms of the mean values as well as standard deviations, in \tabref{symmetric_simulation_numbers}. It can be seen that in most situations, our estimator achieves comparable if not better performance when estimating $\pi_0$ as compared to the other methods for the parameter they are meant to estimate ($\theta$ or $\theta_0$). We also note that our method is significantly influenced by the estimation of $\hat{f}$, therefore in situations where $\hat{f}$ deviates from $f$ often, our estimator will likely result in higher error. In addition, it is clear from the experiments that specifying the center of symmetry is unnecessary.

\begin{table}[htpb]
\centering\small
\caption{Simulated situations for the estimation of a symmetric background component, together with the corresponding values of $\theta$, $\theta_0$, $\pi_0$ (unspecified center), and $\pi_{00}$ (given center), obtained numerically (and rounded at 3 decimals).}
\label{tab:symmetric_simulation_situations}
\bigskip
\setlength{\tabcolsep}{0.03in}
\begin{tabular}{p{0.1\textwidth} 
p{0.525\textwidth} p{0.075\textwidth} p{0.075\textwidth} p{0.075\textwidth} p{0.075\textwidth} 
}
\toprule
{\bf Model} & {\bf Distribution} & {\bf $\theta $} & {\bf $\theta_0$} & {\bf $\pi_0$} & {\bf $\pi_{00}$}\\ 
\midrule
\textsf{S1} & $0.85 \text{ }\mathcal{N} (0, 1) + 0.15 \text{ }\mathcal{N}(3, 1) $ & 0.850 & 0.850 & 0.860 & 0.850 \\ 
\midrule
\textsf{S2} & $0.95 \text{ }\mathcal{N} (0, 1) + 0.05 \text{ }\mathcal{N}(3, 1) $ & 0.950 & 0.950 & 0.950 & 0.950 \\ 
\midrule
\textsf{S3} & $0.85 \text{ }\mathcal{N} (0, 1) + 0.1 \text{ }\mathcal{N}(2.5, 0.75) + 0.05 \text{ }\mathcal{N}(-2.5, 0.75) $ & 0.850 & 0.851 & 0.954 & 0.950 \\
\midrule
\textsf{S4} & $0.85 \text{ }\mathcal{N} (0, 1) + 0.1 \text{ }\mathcal{N}(2.5, 0.75) + 0.05 \text{ }\mathcal{N}(5, 0.75) $ & 0.850 & 0.850 & 0.858 & 0.850\\

\bottomrule
\end{tabular}
\end{table}

\begin{table}[htpb]
\centering\small
\caption{A comparison of various methods for estimating a background component in the situations of \tabref{symmetric_simulation_situations}. For our method, the first and second rows in each situation are for when the center is unspecified, while the third and fourth rows are for when the center is specified to be the origin. Thus $\hat{\pi}_0$ on the first row of each situation is compared with $\pi_0$, $\hat{\pi}_0$ on the third row of each situation is compared with $\pi_{00}$. Otherwise, the $\hat{\theta}_0^{\rm X}$ are compared with $\theta_0$, while the $\hat{\theta}^X$ are compared with $\theta$.}
\label{tab:symmetric_simulation_numbers}
\bigskip
\setlength{\tabcolsep}{0.02in}
\begin{tabular}{p{0.08\textwidth} 
p{0.075\textwidth} 
p{0.075\textwidth} p{0.075\textwidth} p{0.075\textwidth} 
p{0.08\textwidth}
p{0.075\textwidth} p{0.075\textwidth} p{0.075\textwidth} p{0.075\textwidth} p{0.075\textwidth} p{0.075\textwidth}}
\toprule
{\bf Model} & {Center} & {\bf $\hat{\pi}_0$} & {\bf $\hat{\pi}_0^{\rm{L}}$} & {\bf $\hat{\pi}_0^{\rm{U}}$} & {\bf Null} & {\bf $\hat{\theta}_0^{\rm{PSC}}$} & {\bf $\hat{\theta}_0^{\rm{PSH}}$} & {\bf $\hat{\theta}_0^{\rm{PSB}}$} & {\bf $\hat{\theta}^{\rm{E}}$} & {\bf $\hat{\theta}^{\rm{MR}}$} & {\bf $\hat{\theta}^{\rm{CJ}}$}\\ 
\midrule
\textsf{S1} & 0.068 & 0.857 & 0.571 & 1 & $\mathcal{N} (0,1)$ & 0.848 & 0.855 & 0.893 & 0.861 & 0.890 & 0.893 \\ 
\text{   } & (0.052) & (0.021) & (0.059) & (0) & & (0.024) & (0.023) & (0.018) & (0.023) & (0.014) & (0.085)\\ 
 & 0 & 0.835 & 0.561 & 1 &  &  &  & & \\ \text{   } & \textsf{given} & (0.022) & (0.058) & (0) &  &  & &\\
\midrule
\textsf{S2} & 0.024 & 0.936 & 0.631 & 1 & $\mathcal{N}(0,1)$ & 0.938 & 0.954 & 0.985 & 0.955 & 0.972 & 0.963 \\
\text{   } & (0.044) & (0.017) & (0.066) & (0) & & (0.023) & (0.022) & (0.015) & (0.026) & (0.008) & (0.082) \\
 & 0 & 0.925 & 0.623 & 1 &  &  &  &  &   \\ 
\text{   } & \textsf{given} & (0.019) & (0.067) & (0) & &  &  &  & \\
\midrule
\textsf{S3} & 0.046 & 0.945 & 0.597 & 1 & $\mathcal{N} (0,1)$ & 0.864 & 0.942 & 0.937 & 0.856 & 0.896 & 0.707\\ 
\text{   } & (0.050) & (0.019) & (0.063) & (0) & & (0.023) & (0.034) & (0.017) & (0.024) & (0.015) & (0.085)\\ 
 & 0 & 0.930 & 0.587 & 1 &  &  &  & & \\ \text{   } & \textsf{given} & (0.020) & (0.064) & (0) &  &  & &\\
\midrule
\textsf{S4} & 0.070 & 0.856 & 0.574 & 1 & $\mathcal{N}(0,1)$ & 0.846  & 0.854 & 0.891 & 0.849 & 0.889  & 0.713  \\ 
\text{   } & (0.056) & (0.021) & (0.061) & (0) & & (0.023) & (0.024) & (0.018) & (0.024) & (0.014) & (0.088)\\
& 0 & 0.833 & 0.563 & 1   &  &  &  &  &   \\ 
\text{   } & \textsf{given} & (0.022) & (0.060) & (0) & &  &  &  &  \\

\bottomrule
\end{tabular}
\end{table}

\subsection{Real data analysis}
\label{sec:symmetric_read_data}

In this subsection we examine six real datasets where the null component could be reasonably assumed to be symmetric. 

We begin with two datasets where we have sufficient information on the background component. The first one is the Prostate dataset \citep{singh2002gene}, which contains gene expression levels for $n=6033$ genes on $102$ men, $50$ of which are control subjects and $52$ are prostate cancer patients. The main objective is to discover the genes that have a different expression level on the control and prostate patient groups. For each gene, we conduct a two-sided two sample $t$ test on the control subjects and prostate patients, and then transform these $t$ statistics into $z$ values, using
\begin{equation}
    \label{symmetric_t_z_transformation}
    z_i = \Phi^{-1} (F_{100}(t_i)), \text{   } i = 1,2,\dots,6033,
\end{equation}
where $\Phi$ denotes the cdf of the standard normal distribution, and $F_{100}$ denotes the cdf of the $t$ distribution with 100 degrees of freedom. We work with these $n=6033$ $z$ values. From \citep{efron2007size} the background component here could be reasonably assumed to be $\mathcal{N} (0,1)$. The results of the different proportion estimators compared in \secref{symmetric_experiments} are shown in the first row of \tabref{symmetric_realdata}.
The fitted largest symmetric component as well as confidence bands are plotted in \figref{symmetric_realdata_picture_a}. 

Next we consider the Carina dataset \citep{walker2007velocity}, which contains the radial velocities of $n=1266$ stars in Carina, a dwarf spheroidal galaxy, mixed with those of Milky Way stars in the field of view. As \cite{patra} stated, the background distribution of the radial velocity, \textsf{bgstars}, can be acquired from \citep{robin2003synthetic}. The  various estimators are computed and shown in the second row of \tabref{symmetric_realdata}. The fitted largest symmetric component as well as confidence bands are plotted in \figref{symmetric_realdata_picture_b}.

\begin{table}[htpb]
\centering\small
\caption{Two real datasets where background component can be reasonably guessed or derived. We compare the same methods for extracting a background symmetric component as in \secref{symmetric_experiments}. (Note that we work with $z$ values here instead of $p$-values so our results for the Prostate dataset are slightly different from those reported by \cite{patra}.)}
\label{tab:symmetric_realdata}
\bigskip
\setlength{\tabcolsep}{0.035in}
\begin{tabular}{p{0.1\textwidth} 
p{0.07\textwidth} 
p{0.07\textwidth} p{0.07\textwidth} p{0.05\textwidth} 
p{0.08\textwidth}
p{0.07\textwidth} p{0.07\textwidth} p{0.07\textwidth} p{0.07\textwidth}
p{0.07\textwidth} p{0.07\textwidth}}
\toprule
{\bf Model} & {Center} & {\bf $\hat{\pi}_0$} & {\bf $\hat{\pi}_0^{\rm{L}}$} & {\bf $\hat{\pi}_0^{\rm{U}}$} & {\bf Null} & {\bf $\hat{\theta}_0^{\rm{PSC}}$} & {\bf $\hat{\theta}_0^{\rm{PSH}}$} & {\bf $\hat{\theta}_0^{\rm{PSB}}$} & {\bf $\hat{\theta}^{\rm{E}}$} & {\bf $\hat{\theta}^{\rm{MR}}$} & {\bf $\hat{\theta}^{\rm{CJ}}$} \\ 
\midrule
Prostate & 0 & 0.977 & 0.789 & 1 & $\mathcal{N} (0,1)$ & 0.931 & 0.941 & 0.975 & 0.931 & 0.956 & 0.867\\ 
\midrule
Carina & 59 & 0.540 & 0.071 & 1 & \textsf{bgstars} & 0.636 & 0.645 & 0.677 & 0.951 & 0.664 & 0.206 \\  

\bottomrule
\end{tabular}
\end{table}

\begin{figure}[htpb]
    
    \label{fig:symmetric_realdata_known}
    \centering
    \centering
    \subfigure[Prostate dataset]{\label{fig:symmetric_realdata_picture_a}\includegraphics[scale=0.355]{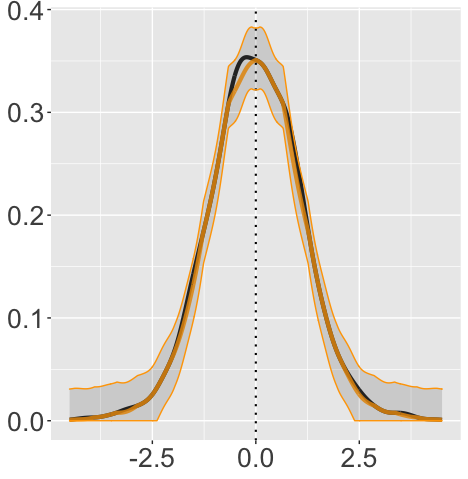}}\qquad
    \centering
    \subfigure[Carina dataset]{\label{fig:symmetric_realdata_picture_b}\includegraphics[scale=0.355]{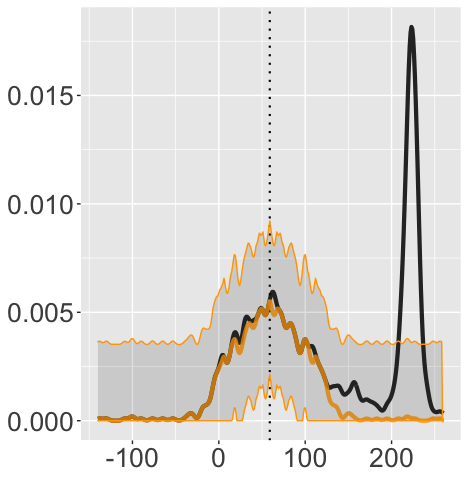}}
    \caption{Estimated symmetric component on the Prostate ($z$ values) and Carina (radial velocity) datasets: the black curve represents the fitted density; the center orange curve represents the computed $\hat{h}_0$; the top and bottom orange curves represent the $95\%$ simultaneous confidence bands for $h_0$; the estimated center of the symmetric component is indicated by a dotted vertical line.}
    \label{fig:prostate_carina}
\end{figure}

Aside from these two real datasets where we know the background distribution, we consider four other real datasets --- three microarray datasets and one police dataset --- where we do not know the null distribution \citep{efron2012large}. Here, out of the methods considered above, only our method and that of \citep{efron2007size} and \citep{cai2010optimal} are applicable. For the first comparison method we use the MLE estimation as presented in \citep[Sec 4]{efron2007size}, which is usually very close to the result of Central Matching as used in \citep{efron2012large}. For the second comparison method we use the estimator in \citep[Sec 3.2]{cai2010optimal}, although we use $\gamma = 0.1$ as the recommended value $\gamma = 0.2$ leads to significant underestimation of the background proportion. Note that both of these methods are still meant to estimate $\theta$.

The HIV dataset \citep{van2003cellular} consists of a study of $4$ HIV subjects and $4$ control subjects. The measurements of $n = 7680$ gene expression levels were acquired using cDNA microarrays on each subject. We compute $t$ statistics for the two sided $t$ test and then transform them into $z$ values using \eqref{symmetric_t_z_transformation}, with the degree of freedom being $6$ here. We would like to know what proportion of these genes do not show a significant difference in expression levels between HIV and control subjects. The results are summarized in the first row of \tabref{symmetric_unknown_realdata}.
The fitted largest symmetric component as well as confidence bands are shown in \figref{symmetric_realdata_picture_c}.

The Leukemia dataset comes from \citep{golub1999molecular}. There are $72$ patients in this study, of which $45$ have ALL (Acute Lymphoblastic Leukemia) and $27$ have AML (Acute Myeloid Leukemia), with AML being considered more severe. High density oligonucleotide microarrays gave expression levels on $n = 7128$ genes. Following \citep[Ch 6.1]{efron2012large}, 
the raw expression levels on each microarray, $x_{i,j}$ for gene $i$ on array $j$, were transformed to a normal score
\begin{equation}
    y_{i,j} = \Phi^{-1} \bigg{(} \big{(}\textsf{rank}(x_{i,j}) - 0.5\big{)} / n\bigg{)},
\end{equation}
where $\textsf{rank}(x_{i,j})$ denotes the rank of $x_{i,j}$ among $n$ raw values of array $j$. 
Then $t$ tests were then conducted on ALL and AML patients, and $t$ statistics were transformed to $z$ values according to \eqref{symmetric_t_z_transformation}, now with $70$ degrees of freedom. As before, we would like to know the proportion of genes that do not show a significant difference in expression levels between ALL and AML patients. The results are summarized in the second row of \tabref{symmetric_unknown_realdata}.
The fitted largest symmetric component as well as confidence bands are shown in \figref{symmetric_realdata_picture_d}.

The Parkinson dataset comes from \citep{lesnick2007genomic}. In this dataset, substantia nigra tissue --- a brain structure located in the mesencephalon that plays an important role in reward, addiction, and movement --- from postmortem the brain of normal and Parkinson disease patients were used for RNA extraction and hybridization, done on Affymetrix microarrays. In this dataset, there are $n = 54 277$ nucleotide sequences whose expression levels were measured on $16$ Parkinson's disease patients and $9$ control patients. We wish to find out the proportion of sequences that do not show significant difference between Parkinson and control patients. The results are summarized in the third row of \tabref{symmetric_unknown_realdata}.
The fitted largest symmetric component as well as confidence bands are shown in \figref{symmetric_realdata_picture_e}.

The Police dataset is analyzed in \citep{ridgeway2009doubly}. In 2006, based on $500 000$ pedestrian stops in New York City, each of the city's $n = 2749$ police officers that were regularly involved in pedestrian stops were assigned a $z$ score on the basis of their stop data, in consideration of possible racial bias. For details on computing this $z$ score, we refer the reader to \citep{ridgeway2009doubly, efron2012large}. Large positive $z$ values are considered as possible evidence of racial bias. We would like to know the percentage of these police officers that do not exhibit a racial bias in pedestrian traffic stops. The estimated proportion are reported on the last row of \tabref{symmetric_unknown_realdata}. 
The symmetric component as well as confidence bands are presented in \figref{symmetric_realdata_picture_f}.

\begin{table}[htpb]
\centering\small
\caption{Real datasets where background distribution is unknown and needs to be estimated. We compare the methods for extracting a background symmetric component among those in \secref{symmetric_experiments} that apply.}
\label{tab:symmetric_unknown_realdata}
\bigskip
\setlength{\tabcolsep}{0.03in}
\begin{tabular}{p{0.15\textwidth} 
p{0.1\textwidth} 
p{0.1\textwidth} p{0.1\textwidth} 
p{0.1\textwidth} p{0.1\textwidth} p{0.1\textwidth} 
}
\toprule
{\bf Model} & {Center} & {\bf $\hat{\pi}_0$} & {\bf $\hat{\pi}_0^{\rm{L}}$} & {\bf $\hat{\pi}_0^{\rm{U}}$} & {\bf $\hat{\theta}^{\rm{E}}$} & {\bf $\hat{\theta}^{\rm{CJ}}$} \\ 
\midrule
HIV & -0.62 & 0.950 & 0.775 & 1 & 0.940 & 0.926      \\  
\midrule
Leukemia & 0.16 & 0.918 & 0.639 & 1 & 0.911 & 0.820       \\  
\midrule
Parkinson & -0.18 & 0.985 & 0.924 & 1 & 0.998 & 0.993  \\  
\midrule
Police & 0.10 & 0.982 & 0.767 & 1 & 0.985 & 0.978       \\  
\bottomrule
\end{tabular}
\end{table}

\begin{figure}[htpb]
    
    \label{fig:symmetric_realdata_unknown}
    \centering

    \centering
    \subfigure[HIV dataset]{\label{fig:symmetric_realdata_picture_c}\includegraphics[scale=0.35]{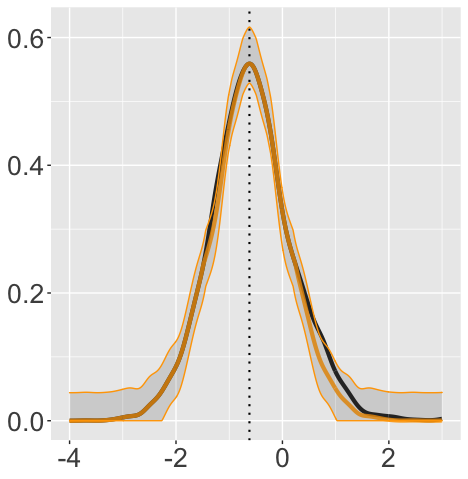}} \qquad
    \centering
    \subfigure[Leukemia dataset]{\label{fig:symmetric_realdata_picture_d}\includegraphics[scale=0.35]{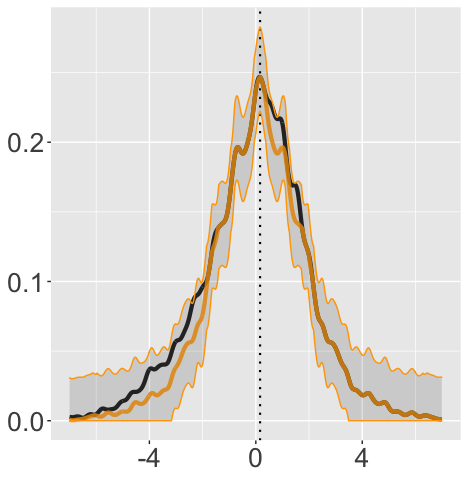}} \\
    \centering
    \subfigure[Parkinson dataset]{\label{fig:symmetric_realdata_picture_e}\includegraphics[scale=0.35]{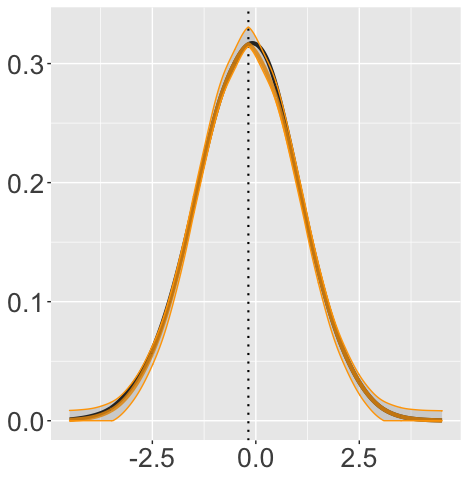}} \qquad
    \centering
    \subfigure[Police dataset]{\label{fig:symmetric_realdata_picture_f}\includegraphics[scale=0.35]{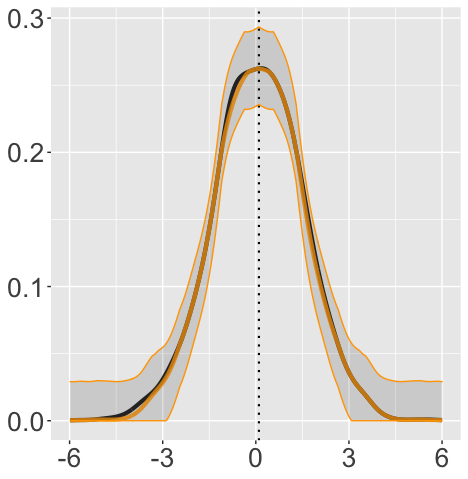}}

    \caption{Estimated symmetric component on HIV ($z$ values), Leukemia ($z$ values), Parkinson ($z$ values), and Police ($z$ scores) datasets: the black curve represents the fitted density; the center orange curve represents the computed $\hat{h}_0$; the top and bottom orange curves represent the $95\%$ simultaneous confidence bands for $h_0$; the estimated center of the symmetric component is indicated by a dotted vertical line.}
    \label{fig:real_data_symmetric_unknown}
\end{figure}

\section{Monotone background component}
\label{sec:monotone}

In this section, we turn our attention to extracting from a density its monotone background component following the Patra--Sen approach. For this to make sense, we only consider densities supported on $\bbR_+ = [0, \infty)$. In fact, all the densities we consider in this section will be supported on $\bbR_+$.
For such a density $f$, we thus define
\begin{equation}
\label{nonincreasing_pi0_definition}
\pi_0 := \sup \big\{\pi: \exists g \in \cM \text{ s.t. } f - \pi g \geq 0 \text{ a.e.}\big\},
\end{equation}
where $\cM$ is the class of monotone (necessarily non-increasing) densities on $\bbR_+$.
Note that $\pi_0 \in [0,1]$ is well-defined for any density $f$, with $\pi_0 = 1$ if and only if $f$ itself is monotone.

Recall that the essential infimum of a measurable set $A$, denoted $\essinf A$, is defined as the supremum over $t \in \bbR$ such that $A \cap (-\infty, t)$ has Lebesgue measure zero.
Everywhere in this section, we will assume that $f$ is c\`adl\`ag, meaning that, at any point, it is continuous from the right and admits a limit from the left.

\begin{thm}
\label{thm:nonincreasing}
Assuming $f$ is c\`adl\`ag, we have 
\begin{align}
\label{nonincreasing_pi0} 
\pi_0 = \int_{0}^\infty h_0(x) d x, 
&& h_0(x) := \essinf\{f(y) : y \leq x\}.
\end{align}
Moreover, if $\pi_0 > 0$ the supremum in \eqref{nonincreasing_pi0_definition} is attained by the following density and no other:
\begin{equation}
\label{nonincreasing_g0} 
g_0(x) := \frac{h_0(x)}{\pi_0}.
\end{equation}
\end{thm}

Note that $\pi_0 = 0$ (i.e., $f$ has no monotone background component) if and only if $\essinf\{f(y) : y \le x\} = 0$ for some $x  >0$, or equivalently, if $\{x \in [0,t] : f(x) \le \eps\}$ has positive measure for all $t > 0$ and all $\eps > 0$. If $f$ is c\`adl\`ag, this condition reduces to $f(0) = 0$. Also, if $f$ is c\`adl\`ag, $h_0(x) = \min\{f(y) : y < x\}$.

\begin{proof}
Note that $\pi_0$ can be equivalently defined as 
\begin{align}
\pi_0 = \sup \big\{\textstyle{\int h}: \text{$h$ is monotone and $0 \le h \le f$ a.e.}\big\}.
\end{align}
Note that $h_0$, as defined in the statement, satisfies the above conditions, implying that $\pi_0 \ge \int h_0$. 
Take $h$ satisfying these same conditions, namely, $h$ is monotone and $0 \le h(x) \le f(x)$ for almost all $x$, say, for $x \in \bbR_+ \setminus A$ where $A$ has Lebesgue measure zero. Take such an $x$. Then for any $y \le x$ we have $h(x) \le h(y)$, and $h(y) \le f(y)$ if in addition $y \notin A$. Hence, 
\[
h(x) \le \inf\{f(y) : y < x, y \notin A\} \le \essinf\{f(y) : y < x\} = h_0(x),
\]
where the second inequality comes from the fact that $A$ has zero Lebesgue measure and the definition of essential infimum.
Hence, $\int h \le \int h_0$ with equality if and only if $h = h_0$ a.e., in particular implying that $\pi_0 \le \int h_0$. We have thus established that $\pi_0 = \int h_0$, and also that $\int h = \pi_0$ if and only if $h = h_0$ a.e.. This not only proves \eqref{nonincreasing_pi0}, but also \eqref{nonincreasing_g0}.  
\end{proof}
   
We have thus established that, in the setting of this section where $f$ is assumed to be c\`adl\`ag, the background component as defined above is given by 
\begin{equation}
h_0(x) = \pi_0 g_0(x) = \essinf\{f(y) : y < x\},
\end{equation}
and $f$ can be expressed as a mixture of the background density and another, unspecified, density $u$, as follows:
\begin{equation}
f = \pi_0 g_0 + (1-\pi_0) u.
\end{equation}
The procedure is summarized in \tabref{monotone_algorithm}.
An illustration of this decomposition is shown in \figref{illustration_nonics}. 
(In this section, $\cE(\sigma)$ denotes the exponential distribution with scale $\sigma$ and $\cG(\kappa, \sigma)$ denotes the Gamma distribution with shape $\kappa$ and scale $\sigma$. Recall that $\cE(\sigma) \equiv \cG(1, \sigma)$.)
By construction, the density $u$ is such that it has no monotone background component in that $\essinf\{u(y) : y < x\} = 0$ for any $x > 0$.

\begin{table}[htpb]
\centering
\caption{Monotone background computation.}
\label{tab:monotone_algorithm}
\bigskip
\setlength{\tabcolsep}{0.22in}
\begin{tabular}{ p{0.9\textwidth}  }
\toprule

{\textbf{inputs}: density ${f}$ defined on $[0, \infty)$} \\ \midrule

$h_0(x) = \essinf\{f(y) : y \leq x\}$

$\pi_0 = \int_0^\infty h_0(x) dx$

$g_0(x) = h_0(x) / \pi_0$ \\

\midrule
\textbf{return} $\pi_0, g_0, h_0$
\\

\bottomrule
\end{tabular}
\end{table}

\begin{figure}[htpb]
    \centering
    \includegraphics[width=0.3\textwidth]{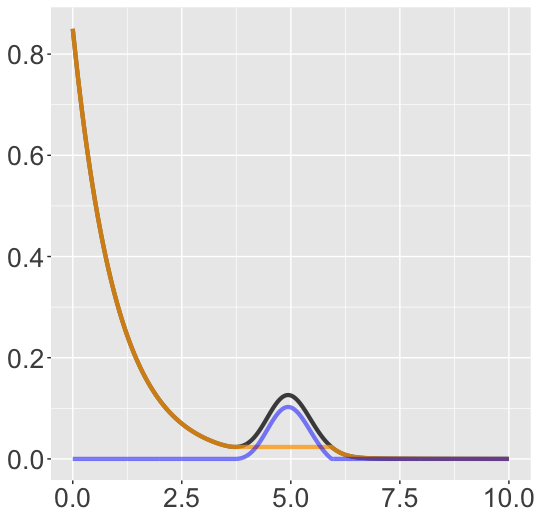}
    \caption{The density $f$ of the Gamma mixture $0.85 \text{ } \mathcal{Exp} (1) + 0.15 \text{ }  \mathcal{Gamma} (100, 1/20)$, in black, and its decomposition into $\pi_0 g_0$, in orange, and $(1 - \pi_0) u$, in blue. Notice that the orange curve has a flat part around the middle and is slightly different from $0.85 \text{ } \mathcal{Exp}(1)$.}
    \label{fig:illustration_nonics}
\end{figure}

\subsection{Estimation and consistency}

In practice, when all we have is a sample of observations, $x_1, \dots, x_n \in \bbR_+$, we first estimate the density, resulting in $\hat f$, and then compute the quantities defined in \eqref{nonincreasing_pi0} and \eqref{nonincreasing_g0} with $\hat f$ in place of $f$. Thus our estimates are
\begin{align}
\hat\pi_0 := \int_{0}^\infty \hat h_0(x) d x, 
&& \hat h_0(x) := \essinf\{\hat f(y) : y < x\},
&& \hat g_0(x) := \frac{\hat h_0(x)}{\hat\pi_0}.
\end{align}

\begin{thm}
Assume $f$ is c\`adl\`ag, and 
suppose that $\hat f$ is a true density, c\`adl\`ag, and is locally uniformly consistent for $f$. Then $\hat h_0$ is locally uniformly consistent for $h_0$ and $\hat \pi_0$ is consistent for $\pi_0$, and if $\pi_0 > 0$, then $\hat g_0$ is locally uniformly consistent fpr $g_0$.
\end{thm}

\begin{proof}
From the definitions,
\begin{align*}
h_0(x) = \essinf\{f(y) : y < x\},
&& \hat h_0(x) := \essinf\{\hat f(y) : y < x\},
\end{align*}
so that
\begin{align*}
|h_0(x) - \hat h_0(x)|
\le \esssup\{|f(y) - \hat f(y)| : y < x\},
\end{align*}
further implying that
\begin{align*}
\esssup_{x < a} |h_0(x) - \hat h_0(x)|
\le \esssup_{y < a} |f(y) - \hat f(y)|, \quad \text{for any $a > 0$},
\end{align*}
from which we get that $\hat h_0$ is locally uniformly consistent for $h_0$ whenever $\hat f$ is locally uniformly consistent for $f$.

For the remaining of the proof, we can follow in the footsteps of the proof of \thmref{symmetric_consistency} based on the fact that, for any $a > 0$,
\[
\left|\int_{[0,a]} h_0 - \int_{[0,a]} \hat h_0\right|
\le \int_{[0,a]} |h_0 - \hat h_0|
\le a \esssup_{[0,a]} |h_0 - \hat h_0|
\to 0,
\]
where the limit is in expectation as the sample size increases.
\end{proof}

\paragraph{Confidence interval and confidence band}

To go beyond point estimators, we suppose that we have available a confidence band for $f$ and deduce from that a confidence interval for $\pi_0$ and a confidence band for $g_0$. 

\begin{thm}
Suppose that we have a confidence band for $f$ as in \eqref{conf1}. Then \eqref{conf2} holds with probability at least $1-\alpha$, where  
\begin{align*}
\hat\pi_l := \int \hat h_l,
&& \hat\pi_u := \int \hat h_u, 
&& \hat h_l(x) := \essinf \{ \hat{f}_l(y) : y < x\},
&& \hat h_u(x) := \essinf \{ \hat{f}_u(y) : y < x\}.
\end{align*}
\end{thm}

The proof is straightforward and thus omitted.

\begin{rem}
So far, we have assumed that the monotone density is supported on $[0, \infty)$, but in principle we can also consider the starting point as unspecified. If this is the case, similar to what we did for the case of a symmetric component in \tabref{symmetric_algorithm}, we can again consider several candidate locations defining the monotone component's support, and select the one yielding the largest monotone component weight.
\end{rem}

\subsection{Numerical experiments}

We are here dealing with densities supported on $[0, \infty)$, and what happens near the origin is completely crucial as is transparent from the definition of $\hat h_0$. It is thus important in practice to choose an estimator for $f$ that behaves well in the vicinity of the origin. As it is well known that kernel density estimators have a substantial bias near the origin, we opted for a different estimator. Many density estimation methods have been proposed to deal with boundary effects, including smoothing splines \citep{gu1993smoothingtheory, gu1993smoothingalgorithm}, local density estimation approaches \citep{fan1993local, loader1996local, hjort1996locally, park2002local}, and local polynomial approximation methods  \citep{cattaneo2020simple, cattaneo2019lpdensity}. For consideration of simplicity and intuitiveness, we consider kernel density estimation using a reflection about the boundary point \citep{schuster1985incorporating, cline1991kernel, karunamuni2005generalized}. We acquire $\hat{f}$ from \citep{schuster1985incorporating} and, as we did before, we acquire a $95\%$ confidence band $[\hat{f}_l, \hat{f}_u]$ from \citep{cheng2019nonparametric}. We also note that $\hat{f}$ is consistent for $f$ as shown by \cite{schuster1985incorporating}.

We consider two different situations as listed in \tabref{monotone_simulation_situations}. Each situation's corresponding $\theta$, $\theta_0$, and $\pi_0$, defined as in \eqref{nonincreasing_pi0_definition}, are also presented. We again generate a sample of size $n = 1000$ from each model, and repeat each setting $1000$ times. The mean values as well as standard deviations of our method and related methods are reported in \tabref{nonincreasing_simulation_numbers}. 

In the situation \textsf{M1}, our estimator achieves a smaller estimation error for $\pi_0$ than all other methods for their corresponding target, either $\theta$ or $\theta_0$, even with much fewer information on the background component. In situation \textsf{M2}, our method has a slightly higher error than the error of $\hat{\theta}_0^{\rm{PSH}}$, $\hat{\theta}_0^{\rm{E}}$, both having complete information on the background component, but lower than that of $\hat{\theta}_0^{\rm{PSC}}$ and $\hat{\theta}_0^{\rm{CJ}}$.

\begin{table}[htpb]
\centering\small
\caption{Simulated situations for the estimation of a monotone background component, together with the corresponding values of $\theta$, $\theta_0$, $\pi_0$, obtained numerically (and rounded at 3 decimals).}
\label{tab:monotone_simulation_situations}
\bigskip
\setlength{\tabcolsep}{0.05in}
\begin{tabular}{p{0.1\textwidth} 
p{0.3\textwidth} p{0.1\textwidth} p{0.1\textwidth} p{0.1\textwidth} 
}
\toprule
{\bf Model} & {\bf Distribution} & {\bf $\theta $} & {\bf $\theta_0 $}& {\bf $\pi_0$}\\ 
\midrule
\textsf{M1} & $0.85 \text{ }\mathcal{E} (1) + 0.15 \text{ } \mathcal{G} (50, 1/10) $ & 0.850 & 0.850 & 0.922  \\ 
\midrule
\textsf{M2} & $0.95 \text{ }\mathcal{E} (1) + 0.05 \text{ } \mathcal{G} (50, 1/10) $ & 0.950 & 0.950 & 0.993 \\

\bottomrule
\end{tabular}
\end{table}

\begin{table}[htpb]
\centering\small
\caption{A comparison of various methods for estimating a monotone background component in the situations of \tabref{monotone_simulation_situations}. As always, $\hat\pi_0$ is compared with $\pi_0$, the $\hat{\theta}_0^{\rm X}$ are compared with $\theta_0$, while the $\hat{\theta}^X$ are compared with $\theta$.}
\label{tab:nonincreasing_simulation_numbers}
\bigskip
\setlength{\tabcolsep}{0.025in}
\begin{tabular}{p{0.1\textwidth} 
p{0.075\textwidth} p{0.075\textwidth} p{0.075\textwidth} 
p{0.1\textwidth}
p{0.075\textwidth} p{0.075\textwidth} p{0.075\textwidth} p{0.075\textwidth}
p{0.075\textwidth}
p{0.075\textwidth}
}
\toprule
{\bf Model} & {\bf $\hat{\pi}_0$} & {\bf $\hat{\pi}_0^{\rm{L}}$} & {\bf $\hat{\pi}_0^{\rm{U}}$} & {\bf Null} & {\bf $\hat{\theta}_0^{\rm{PSC}}$} & {\bf $\hat{\theta}_0^{\rm{PSH}}$} & {\bf $\hat{\theta}_0^{\rm{PSB}}$} & {\bf $\hat{\theta}^{\rm{E}}$} 
& {\bf $\hat{\theta}^{\rm{MR}}$}
& {\bf $\hat{\theta}^{\rm{CJ}}$}\\ \midrule
\textsf{M1}   & 0.920 & 0.460 & 1 & $\mathcal{Exp}(1)$ & 0.843  & 0.854 & 0.889 & 0.841 & 0.866 & 0.533 \\ 
\text{   }   & (0.020) & (0.053) & (0) & & (0.022) & (0.022) & (0.018) & (0.028) & (0.013) & (0.085)\\ 
\midrule
\textsf{M2}   & 0.984 & 0.519 & 1 & $\mathcal{Exp}(1)$ & 0.936 & 0.953 & 0.984 & 0.954 & 0.964 & 0.842\\ 
\text{   }   & (0.012) & (0.061)  & (0)  & & (0.022) & (0.020) & (0.014) & (0.028) & (0.009) & (0.083)\\

\bottomrule
\end{tabular}
\end{table}

\subsection{Real data analysis}
\label{sec:monotone_real_data}

In this subsection we consider a real dataset where the background component could be assumed to be monotonic nonincreasing. We look at the Coronavirus dataset \citep{coronavirus}, acquired from the \textit{COVID-19 Data Repository by the Center for Systems Science and Engineering (CSSE)} at Johns Hopkins University\footnote{~\url{https://github.com/CSSEGISandData/COVID-19}}. It is well known that new coronavirus cases are consistently decreasing in the USA currently, and this trend could be seen to begin on Jan 8, 2021 as shown by the New York Times interactive data\footnote{~\url{https://www.nytimes.com/interactive/2021/us/covid-cases.html}}. For each person infected on or after Jan 8, 2021, we count the number of days between that day and the time they were infected. We are interested in quantifying how monotonic that downward trend in coronavirus infections is. 
As we do not know the actual background distribution here, and Gaussian distributions are of not particular relevance, we find that none of the other comparison methods in \secref{symmetric_experiments} are applicable, and therefore only provide our method's estimate in \tabref{monotone_coronavirus_table} and \figref{monotone_coronavirus_picture}. Numerically, it can be seen that the background monotonic component accounts for around $96.7\%$ of the new cases arising on or after Jan 8, 2021.

\begin{table}[htpb]
\centering\small
\caption{Coronavirus dataset where it is of interest to gauge how monotonic the trend is starting in Jan 8, 2021.}
\label{tab:monotone_coronavirus_table}
\bigskip
\setlength{\tabcolsep}{0.1in}

\begin{tabular}{p{0.125\textwidth} 
p{0.1\textwidth} p{0.1\textwidth} p{0.1\textwidth} 
}
\toprule
{\bf Model} & {\bf $\hat{\pi}_0$} & {\bf $\hat{\pi}_0^{\rm{L}}$} & {\bf $\hat{\pi}_0^{\rm{U}}$} \\ \midrule
\textsf{Coronavirus}   & 0.967 & 0.955 & 0.968  \\ 

\bottomrule
\end{tabular}
\end{table}

\begin{figure}[htpb]
    \centering
    \includegraphics[width=0.3\textwidth]{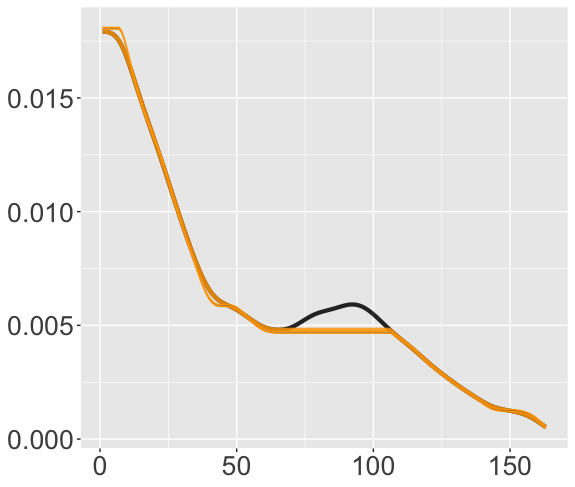}
    \caption{Estimated background monotone component on the Coronavirus dataset: the black curve represents the fitted density; the center orange curve represents the computed $\hat{h}_0$; and the top and bottom orange curves represent the 95\% simultaneous confidence bands for $h_0$. (Due to the large amount of data here the confidence bands are very narrow.)}
    \label{fig:monotone_coronavirus_picture}
\end{figure}

\section{Log-concave background component}
\label{sec:log-concave}

Our last emblematic setting is that of extracting a log-concave background component from a density. Log-concave densities have been widely studied in the literature \citep{walther2009inference, samworth2018recent}, and include a wide variety of distributions, including all Gaussian densities, all Laplace densities, all exponential densities, and all uniform densities. 
They have been extensively used in mixture models \citep{chang2007clustering, hu2016maximum, pu2020algorithm}.

Following \cite{patra}, for a density $f$ we define
\begin{equation}
\label{logconcave_pi0_definition}
\pi_0 := \sup \big\{\pi: \exists g \in \cC \text{ s.t. } f - \pi g \geq 0 \text{ a.e.}\big\},
\end{equation}
where $\cC$ is the class of log-concave densities.
Note that $\pi_0 \in [0,1]$, with $\pi_0 = 1$ if and only if $f$ itself is log-concave.

\begin{thm}
$\pi_0$ is the value of following optimization problem:
\begin{equation}
\begin{aligned}
\label{logconcave_original_maximization}
\text{maximize}& \quad \int_{-\infty}^{\infty} h(x) dx \\
\text{over}& \quad \big\{ h: \mathbb{R} \xrightarrow{} \mathbb{R}_+,\ \text{log-concave},\ h \leq f\big\}.
\end{aligned}
\end{equation}
Indeed, this problem admits a solution, although it may not be unique.
\end{thm}

\begin{proof}
From the definition in \eqref{logconcave_pi0_definition}, it is clear that
\begin{equation}
\pi_0 = \sup \big\{\textstyle{\int h}: \text{$h$ is log-concave and $0 \le h \le f$ a.e.}\big\}.
\end{equation}
Thus we only need to show that the problem \eqref{logconcave_original_maximization} admits a solution --- and then show that it may not be unique to complete the proof.
Here the arguments are a bit more involved.

Let $(h_k)$ be a solution sequence to the problem \eqref{logconcave_original_maximization}, meaning that $h_k : \bbR \to \bbR_+$ is log-concave and satisfies $h_k \le f$, and that $q_k := \int h_k$ converges, as $k \to \infty$, to the value of the optimization problem \eqref{logconcave_original_maximization}, denoted $q_*$ henceforth. Note that $0 \le q_k \le 1$ since $0 \le h_k \le f$ and $\int f = 1$, implying that $0 \le q_* \le 1$.  We only need to consider the case where $q_* > 0$, for when $q_* = 0$ the constant function $h \equiv 0$ is a solution. Without loss of generality, we assume that $q_k > 0$ for all $k$. 

Each $h_k$ is log-concave, and from this we know the following: its support is an interval, which we denote $[a_k, b_k]$ with $-\infty \le a_k < b_k \le \infty$; $h_k$ is continuous and strictly positive on $(a_k, b_k)$; and $x \mapsto (\log h_k(y) - \log h_k(x))/(y-x)$ is non-increasing in $x \le y$ and $y \mapsto (\log h_k(y) - \log h_k(x))/(y-x)$ is non-increasing in $y \ge x$. 
Extracting a subsequence if needed, we assume without loss of generality that $a_k \to a$ and $b_k \to b$ as $k \to \infty$, for some $-\infty \le a \le b \le \infty$.

Define $F(t) = \int_{-\infty}^t f(x) dx$ and $H_k(t) = \int_{-\infty}^t h_k(x) dx$.
We have that $H_k$ is a non-decreasing, and if $s \le t$, $H_k(t) - H_k(s) \le F(t) - F(s)$, because $h_k \le f$.
In particular, by Helly's selection theorem (equivalent to Prokhorov's theorem when dealing with distribution functions), we can extract a subsequence that converges pointwise. Without loss of generality, we assume that $(H_k)$ itself converges, and let $H$ denote its limit.
Note that $H$ is constant outside of $[a,b]$.
In fact, $H(x) = 0$ when $x < a$, because $x < a_k$ eventually, forcing $H_k(x) = 0$. For $x, x' > b$, we have that $x, x' > b_k$ for large enough $k$, implying that $H_k(x) = H_k(x')$, yielding $H(x) = H(x')$ by taking the limit $k \to \infty$.
Note also that, for all $s \le t$, 
\[
H(t) - H(s)
= \lim_{k \to \infty} (H_k(t) - H_k(s))
\le F(t) - F(s).
\]
This implies that $H$ is absolutely continuous with derivative, denoted $h$, satisfying $h \le f$ a.e..  

We claim that $h$ is a solution to \eqref{logconcave_original_maximization}. 
We already know that $0 \le h \le f$ a.e.. In addition, we also have $\int_{-\infty}^\infty h \ge q_*$.
To see this, fix $\eps > 0$, and let $t$ be large enough that $\int_t^\infty f \le \eps$.
Because $h_k \le f$, we have $H_k(t) = q_k - \int_t^\infty h_k \ge q_k - \eps$, implying that $H(t) \ge q_* -\eps$ by taking the limit as $k \to\infty$. Hence, $\int_{-\infty}^\infty h = H(\infty) \ge H(t) \ge q^* -\eps$, and $\eps > 0$ being otherwise arbitrary, we deduce that $\int_{-\infty}^\infty h \ge q_*$.
It thus remains to show that $h$ is log-concave. 

We establish this claim by proving that, extracting a subsequence if needed, $h_k$ converges to $h$ a.e., and it is enough to do so in an interval. Thus let $x_0$ and $\Delta > 0$ be such that $[x_0-4\Delta, x_0+4\Delta] \subset (a, b)$. Note that $H$ is strictly increasing on $(a, b)$, because each $H_k$ is strictly increasing on $(a_k, b_k)$ due to $h_k$ being log-concave.
Take any $x_0-\Delta \le x < y \le x_0+\Delta$ and let $\delta_k = (\log h_k(y) - \log h_k(x))/(y-x)$. 
We assume, for example, that $\delta_k \ge 0$, and bound it from above.
Let $z_k \in [x_0-4\Delta, x_0-3\Delta]$ be such that $H_k(x_0-3\Delta) - H_k(x_0-4\Delta) = h_k(z_k) \Delta$, which exists by the mean-value theorem. Note that $h_k(z_k) \Delta \to \Delta_1 := H(x_0-3\Delta) - H(x_0-4\Delta)$, so that $h_k(z_k) \ge \Delta_2 := \Delta_1/2\Delta > 0$, eventually.
Now, for any $z$ in $[x_0-2\Delta, x_0-\Delta]$, due to $z_k < z < x < y$ and $\log h_k$ being concave,
\[
\frac{\log h_k(z) - \log h_k(z_k)}{z-z_k} \ge \frac{\log h_k(y) - \log h_k(x)}{y-x} = \delta_k,
\] 
which implies
\[
h_k(z) 
\ge h_k(z_k) \exp(\delta_k (z-z_k))
\ge \Delta_2  \exp(\delta_k \Delta).
\]
This being true for all such $z$, we have
\[
1
\ge \int_{-\infty}^\infty f(z) \d z
\ge \int_{-\infty}^\infty h_k(z) \d z 
\ge \int_{x_0-2\Delta}^{x_0-\Delta} h_k(z) \d z
\ge \Delta \cdot \Delta_2  \exp(\delta_k \Delta),
\]
allowing us to derive $\delta_k \le M_1 := \Delta^{-1} \log(2/\Delta_1)$.
We can deal with the case where $\delta_k < 0$ by symmetry, obtaining that, for all $k$ sufficiently large and all $x_0-\Delta \le x < y \le x_0+\Delta$, 
\[
\left|\frac{\log h_k(y) - \log h_k(x)}{y-x}\right| \le M_1.
\]
Let $u_k = h_k(x_0)$. Because $h_k$ is unimodal, either $h_k(x) \le u_k$ for all $x \le x_0$ or $h_k(x) \le u_k$ for all $x \ge x_0$. Extracting a subsequence if needed, and by symmetry, assume that the former is true for all $k$ large enough. Then 
\begin{align*}
\Delta_1
&= H(x_0-3\Delta) - H(x_0-4\Delta) \\
&= \lim_{k \to \infty} (H_k(x_0-3\Delta) - H_k(x_0-4\Delta)) \\
&= \lim_{k \to \infty} \int_{x_0-3\Delta}^{x_0-4\Delta} h_k(x) \d x \\
&\le \liminf_{k \to \infty} \Delta u_k, 
\end{align*}
so that $u_k \ge \Delta_2 > 0$, eventually.
Assuming so, we have $u_k h_k(x_0) \ge \Delta_1/\Delta \ge \Delta_2$.
And we also have $h_k(x_0) \le f(x_0)$, and together, $|\log h_k(x_0)| \le M_2 := |\log (\Delta_2)| \vee |\log f(x_0)|$.
With the triangle inequality, we thus have, for all $x \in [x_0-\Delta, x_0+\Delta]$,
\[
|\log h_k(x)| \le M_1 |x-x_0| + |\log h_k(x_0)| \le M_1 \Delta + M_2 =: M_3.
\]
The family of functions $(\log h_k)$ (starting at $k$ large enough) is thus uniformly bounded and equicontinuous on $[x_0-\Delta, x_0+\Delta]$, so that by the Arzel\`a--Ascoli theorem, we have that $(\log h_k)$ is precompact for the uniform convergence on that interval. Therefore, the same is true for $(h_k)$. Let $h_\infty$ be the uniform limit of a subsequence, and note that $h_\infty$ is continuous. 
$h_\infty$ must also be a weak limit as well, since uniform convergence on a compact interval implies weak convergence on that interval.
Therefore $h_\infty = h$ a.e.~on that interval, since $(h_k)$ converges weakly to $h$.
Hence, all the uniform limits of $(h_k)$ must coincide with $h$ a.e., and since any such limit must be continuous, it means that they are the same. We conclude that, on the interval under consideration, $h$ is equal a.e.~to a continuous function which is the (only) uniform limit of $(h_k)$, and in particular, $(h_k)$ converges pointwise a.e.~to $h$ on that interval. 

Thus, we have proved that the optimization problem \eqref{logconcave_original_maximization} has at least one solution. We now show that there may be multiple solutions. This is the case, for instance, when $f = \frac1m f_1 + \cdots + \frac1m f_m$ where each $f_j$ is a log-concave density and these densities have support sets that are pairwise disjoint. In that case, any of the components, meaning any $f_j$, is a solution to \eqref{logconcave_original_maximization}, and these are the only solutions. This comes from the fact that the support of a log-concave distribution is necessarily an interval.
\end{proof}

We have thus established that a density has at least one log-concave background component, and possibly multiple ones, corresponding to the solutions to \eqref{logconcave_original_maximization}. If $h$ is one such solution, then $\pi_0 = \int h$ and we may define the corresponding density as $g = h/\pi_0$ if $\pi_0 > 0$.
Then $f$ can be expressed as a mixture of the background density $g$ and another, unspecified, density $u$, as follows
\begin{equation}
f = \pi_0 g + (1-\pi_0) u.
\end{equation}
(We do not use the notation $h_0$ and $g_0$ here, since these may not be uniquely defined.)
The procedure is summarized in \tabref{logconcave_algorithm}. 
An illustration of this decomposition is shown in \figref{illustration_logconcave}. Note that the density $u$ may have a non-trivial log-concave background component. This is the case, for example, if $f$ is the (nontrivial) mixture of two log-concave densities with disjoint support sets, in which case $u$ is one of these log-concave densities.

\begin{table}[htpb]
\centering
\caption{Log-concave background computation. (The input $d$ is used to initialize the function $v$ --- discretized as $\bv$ below --- that is bounded from above by $\log f$. In our experiments, we chose $d = 0.02$.) }
\label{tab:logconcave_algorithm}
\bigskip
\setlength{\tabcolsep}{0.22in}
\begin{tabular}{ p{0.9\textwidth}  }
\toprule
{\textbf{inputs}: equally spaced gridpoints $\mathbf{t} = \{t_1,t_2,\dots, t_k\}$, density $f$, a boolean $R$ indicating whether to use the Riemann integral approximation, initialization amount $d$} \\ \midrule

initialize $\mathbf{w} = (0, 0, \dots, 0)$ with length $k$

\textbf{for} $i=1,\dots,k$ \textbf{do}

\hspace{3mm} $\mathbf{w}[i] = \log(f(t_i))$

compute $\mathbf{A}$ and $\mathbf{b}$ as in \eqref{pointwise_approximation_a_b}

initialize $\mathbf{v} = \mathbf{w} - (d, d, \dots, d)$ 

\textbf{if} $R$ = \textsf{true} \textbf{then}

\hspace{3mm} do optimization \eqref{riemann_linear_approximation} using SQP, and record the optimizer $\mathbf{v}$ and the maximum as $\pi_0$

\textbf{else}

\hspace{3mm} do optimization \eqref{actual_linear_approximation} using SQP, and record the optimizer $\mathbf{v}$ and the maximum as $\pi_0$ \\

\midrule
\textbf{return} $\mathbf{v}, \pi_0$
\\

\bottomrule
\end{tabular}
\end{table}

\begin{figure}[htpb]
    \centering
    \includegraphics[width=0.3\textwidth]{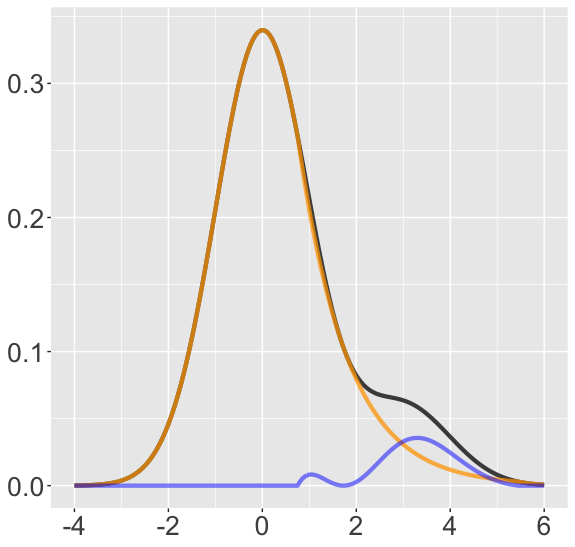}
    \caption{The density $f$ of the Gaussian mixture $0.85 \text{ } \mathcal{N}(0, 1) + 0.15 \text{ } \mathcal{N} (3, 1)$, in black, and its decomposition into $\pi_0 g$, in orange, and $(1 - \pi_0) u$, in blue. The orange curve is acquired by maximizing \eqref{not_apx_area}, but maximizing the Riemann integral approximation \eqref{riemann_linear_approximation} gives almost the same result (with the maximum difference between the two curves being $2 \times 10^{-14} $). Notice $\pi_0 = 0.931$ as $\pi_0 g$ also takes a non-negligible weight from the smaller component $0.15 \text{ } \mathcal{N} (0, 1)$.  }
    \label{fig:illustration_logconcave}
\end{figure}

\subsection{Estimation and consistency}

In practice, based on a sample, we estimate $f$, resulting in $\hat f$, and simply obtain estimates by plug-in as we did before, this time via
\begin{equation}
\begin{aligned}
\label{logconcave_fitted_maximization}
\text{maximize}& \quad \int_{-\infty}^{\infty} h(x) dx \\
\text{over}& \quad \big\{ h: \mathbb{R} \xrightarrow{} \mathbb{R}_+,\ \text{log-concave},\ h \leq \hat f\big\}.
\end{aligned}
\end{equation}

At least formally, let $\hat h$ be a solution to \eqref{logconcave_fitted_maximization}. 
We then define
\begin{align}
\hat{\pi}_0 = \int_{-\infty}^{\infty} \hat h(x) dx,
&&\hat{g}(x) = \frac{\hat h(x)}{\hat\pi_0}.
\end{align}

Here, to avoid technicalities, we work under the assumption that the density estimator $\hat f = \hat f_n$ satisfies
\begin{equation}
\label{logconcave_consistency_condition}
\E\left[\esssup_{|x| \le a} \frac{\max\{\hat f_n(x), f(x)\}}{\min\{\hat f_n(x), f(x)\}}\right] \to 1, \quad n \to \infty, \quad \forall a > 0.
\end{equation}
This condition is better suited for when the density $f$ approaches 0 only at infinity.
Also for the sake of simplicity, we only establish consistency for $\hat\pi_0$.

\begin{thm}
When \eqref{logconcave_consistency_condition} holds, $\hat\pi_0$ is a consistent.
\end{thm}

\begin{proof}
Fix $\eta > 0$ and $a > 0$, and consider the event, denoted $\Omega$, that 
\begin{equation}
\esssup_{|x| \le a} \frac{\max\{\hat f(x), f(x)\}}{\min\{\hat f(x), f(x)\}} \le 1+\eta.
\end{equation}
Because of \eqref{logconcave_consistency_condition}, $\Omega$ happens with probability tending to~1 as the sample size increases.
Let $\eps = 1 - \int_{[-a,a]} f$, which is small when $a$ is large.

Assume that $\Omega$ holds. Let $h$ be a solution to \eqref{logconcave_original_maximization}, so that $\pi_0 = \int h$. Then define $\tilde h(x) = (1-\eta) h(x) \IND{|x| \le a}$ and note that 
\[
\int \tilde h = (1-\eta) \int_{[-a,a]} h \ge (1-\eta)(\pi_0 - \eps),
\]
so that $\int \tilde h$ is close to $\pi_0$ when $\eta$ is small and $a$ is large.
We also note that $\tilde h$ is log-concave and satisfies $0 \le \tilde h \le (1-\eta) f$ a.e.. Under $\Omega$, we have $f \le (1+\eta) \hat f$, and so it is also the case that $\tilde h \le (1-\eta)(1+\eps)\hat f \le \hat f$, assuming as we do that $\eta$ and $\eps$ are small enough. Then $\int \tilde h \le \hat \pi_0$, by definition of $\hat \pi_0$. Gathering everything, we obtain that $(1-\eta)(\pi_0 - \eps) \le \hat \pi_0$. By letting $\eta \to 0$ and $a \to \infty$ so that $\eps \to 0$, we have established that $\liminf\hat\pi_0 \ge \pi_0$ in probability. (The $\liminf$ needs to be understood as the sample size increases.)

The reverse relation, meaning $\limsup \hat\pi_0 \le \pi_0$, can be derived similarly starting with a solution $\hat h$ to \eqref{logconcave_fitted_maximization}.
\end{proof}

\paragraph{Confidence interval}
Once again, if we have available a confidence band for $f$, we can deduce from that a confidence interval for $\pi_0$. 

\begin{thm}
Suppose that we have a confidence band for $f$ as in \eqref{conf1}. Then 
\begin{align}
\hat\pi_l \le \pi_0 \le \hat\pi_u
\end{align}
holds with probability at least $1-\alpha$, where $\hat\pi_l$ and $\hat\pi_u$ are the values of the optimization problem \eqref{logconcave_original_maximization} with $f$ replaced by $\hat f_l$ and $\hat f_u$, respectively. 
\end{thm}

The proof is straightforward and thus omitted.

\subsection{Numerical method}

Unlike the previous sections, here the computation of our estimator(s) is non-trivial: indeed, after computing $\hat f$ with an off-the-shelf procedure, we need to solve the optimization problem \eqref{logconcave_fitted_maximization}. Thus, in this section, we discuss how to solve this optimization problem.

Although least concave majorants (or equivalently greatest convex minorants) have been considered extensively in the literature, for example in \citep{francuu2017new, jongbloed1998iterative}, the problem \eqref{logconcave_original_maximization} calls for a type of greatest concave minorant, and we were not able to find references in the literature that directly tackle this problem. We do want to mention \citep{gorokhovik2018minimal}, where a similar concept is discussed, but the definition is different from ours and no numerical procedure to solve the problem is provided. For lack of structure to exploit, we propose a direct discretization followed by an application of sequential quadratic programming (SQP), first proposed by \cite{wilson1963simplicial}. For more details on SQP, we point the reader to \citep{gill2012sequential} or \citep[Ch 18]{nocedal2006sequential}.

Going back to \eqref{logconcave_original_maximization}, where here $f$ plays the role of a generic density on the real line, the main idea is to restrict $v := \log h$ to be a continuous, concave, piecewise linear function. Once discretized, the integral admits a simple closed-form expression and the concavity constraint is transformed into a set of linear inequality constraints. 

To setup the discretization, for $k \ge 1$ integer, let $t_{-k, k} < t_{-k+1, k} < \cdots < t_{k-1, k} < t_{k,k}$ be such that $t_{j,k} = - t_{-j,k}$ (symmetry) and $t_{j+1, k} - t_{j,k} = \delta_k$  (equispaced) for all $j$, with $\delta_k \to 0$ (dense) and $t_{k,k} \to \infty$ as $k\to\infty$ (spanning the real line). Suppose that $v$ is concave with $v \le \log f$ and to that function associate the triangular sequence $v_{j,k} := v(t_{j,k})$. Then, for each $k$,
\begin{align*}
-v_{j+1,k} + 2 v_{j,k} - v_{j-1,k}
&= - v(t_{j+1,k}) + 2 v(t_{j,k}) - v(t_{j-1,k}) \\
&= 2 \big(v(t_{j,k}) - \tfrac12 v(t_{j+1,k}+\delta_k) - \tfrac12 v(t_{j-1,k}+\delta_k)\big)
\ge 0, \quad \text{for all } k,
\end{align*}
by the fact that $v$ is concave.
In addition, $v_{j,k} = v(t_{j,k}) \le \log f(t_{j,k}) =: u_{j,k}$ (which are given).
Instead of working directly with a generic concave function $v$, we work with those that are piecewise linear as they are uniquely determined by their values at the grid points if we further restrict them to be $=-\infty$ on $(-\infty, t_{-k,k}) \cup (t_{k,k},+\infty)$. Effectively, at $k$, we replace in \eqref{logconcave_original_maximization} the class $\cC$ with the class $\cC_k$ of functions $h$ such that $v = \log h$ is log-concave, linear on each interval $[t_{j,k}, t_{j+1,k}]$, $v(x) = -\infty$ for $x < t_{-k,k}$ or $x > t_{k,k}$, and that satisfies $v(t_{j,k}) \le \log f(t_{j,k})$, for all $j$.

This leads us to the following optimization problem, which instead of being over a function space is over a Euclidean space:
\begin{equation}
    \label{actual_linear_approximation}
\begin{split}
\text{maximize} \quad & \Lambda(\bv) \\
\text{over} \quad & \mathbf{v} = [v_{-k,k}, \dots, v_{k,k}]^\top \quad \text{such that} \quad 
\mathbf{A} \mathbf{v} \geq \mathbf{b},
\end{split}
\end{equation}
where 
\begin{equation}
\label{not_apx_area}
\Lambda(v_{-k}, \dots, v_k) := \delta_k \sum_{j=-k}^{k-1} \lambda(v_j, v_{j+1}),
\end{equation}
\begin{equation}
\lambda(x, y) := \mathds{1} \{x \neq y\} \frac{\exp(x) - \exp(y)}{x-y} + \mathds{1} \{x=y\} \exp(x),
\end{equation}
and where
\begin{equation}
    \label{pointwise_approximation_a_b}
    \mathbf{A} = \begin{bmatrix} -1 & 2 & -1& 0 & \dots & 0 & 0 & 0\\ 0 & -1 & 2 & -1 & \dots &0 & 0 & 0\\ \vdots&\vdots&\vdots&\vdots&\vdots&\vdots& \vdots& \vdots \\ 0 & 0 & 0 & 0 &\dots & -1 &2 & -1 \\ -1 & 0 & 0 & 0 &\dots & 0 & 0 & 0 \\ 0 & -1 & 0 & 0 &\dots & 0 & 0 & 0 \\ \vdots&\vdots&\vdots&\vdots&\vdots& \vdots&\vdots& \vdots \\ 0 & 0 & 0 & 0 &\dots & 0 & 0 & -1
    \end{bmatrix}, 
    \qquad \mathbf{b} = \begin{bmatrix} 0 \\ 0 \\ \vdots \\ 0 \\ -u_1 \\ -u_2 \\ \vdots \\ -u_k  \end{bmatrix}\, .
\end{equation}
As we are not aware of any method that could solve \eqref{actual_linear_approximation} exactly, we will use sequential quadratic programming (SQP). In our implementation we use the R package {\sf nloptr}\,\footnote{~\url{https://cran.r-project.org/web/packages/nloptr/index.html}} based on an original implementation of \cite{kraft1988software}.  


\begin{thm}
Assume that $f$ is Lipschitz and locally bounded away from~0.
Then the value of discretized optimization problem \eqref{actual_linear_approximation} converges, as $k \to \infty$, to the value of the original optimization problem \eqref{logconcave_original_maximization}. 
\end{thm}

The assumptions made on $f$ are really for convenience --- to expedite the proof of the result while still including interesting situations --- and we do expect that the result holds more broadly. We also note that, in our workflow, the problem \eqref{actual_linear_approximation} is solved for $\hat f$, and that $\hat f$ satisfies these conditions (remember that the sample size is held fixed here) if it is a kernel density estimator based on a smooth kernel function supported on the entire real line like the Gaussian kernel. 

\begin{proof}
Let $h$ be a solution to \eqref{logconcave_original_maximization} so that $\int h = \pi_0$, where $\pi_0$ denotes the value of \eqref{logconcave_original_maximization}.
Define $v = \log h$ and let $v_{j,k} = v(t_{j,k})$. As we explained above, this makes $\bv_k := (v_{-k,k}, \dots, v_{k,k})$ feasible for \eqref{actual_linear_approximation}. Therefore $\Lambda(\bv_k) \le \pi_{0,k}$, where $\pi_{0,k}$ denotes the value of \eqref{actual_linear_approximation}. 
On the other hand, let $\tilde v_k$ denote the piecewise linear approximation to $v$ on the grid, meaning that $\tilde v_k(x) = -\infty$ if $x < t_{-k,k}$ or $x > t_{k,k}$, $\tilde v_k(t_{j,k}) = v(t_{j,k})$ and $\tilde v_k$ linear on $[t_{j,k}, t_{j+1,k}]$ for all $j$. We then have 
\begin{equation}
\Lambda(\bv_k)
= \int \tilde h_k
\xrightarrow{k \to \infty} \int h,
\end{equation}
where the convergence is justified, for example, by the fact that $\tilde h_k \to h$ pointwise and $0 \le \tilde h_k \le h$ (because, by concavity of $v$, $v \ge \tilde v_k$), so that the dominated convergence theorem applies.
From this we get that 
\begin{equation}
\liminf_{k\to\infty} \pi_{0,k} \ge \pi_0.
\end{equation}

In the other direction, let $\bv_k$ be a solution of \eqref{actual_linear_approximation}, so that $\Lambda(\bv_k) = \pi_{0,k}$. Let $\tilde v_k$ denote the linear interpolation of $\bv_k$ on the grid $(t_{-k,k}, \dots, t_{k,k})$, defined exactly as done above, and let $\tilde h_k = \exp(\tilde v_k)$. We have that $\tilde h_k \le f$ at the grid points, but not necessarily elsewhere. To work in that direction, fix $a > 0$, taken arbitrarily large in what follows. Because of our assumptions on $f$, we have that $u := \log f$ is Lipschitz on $[-a,a]$, say with Lipschitz constant $L$. 
In particular, with $\tilde u_k$ denoting the linear approximation of $u$ based on the grid, we have 
\begin{equation}
|u(x) - \tilde u_k(x)|
\le L \delta_k =: \eta_k.
\end{equation} 
Define $\bar v_k(x) = \tilde v_k(x) -\eta_k$ if $x \in [-a,a]$ and $\bar v_k(x) = -\infty$ otherwise. Note that $\bar v_k$ is also concave and piecewise linear, and because $\tilde h_k \le f$, $\tilde v_k \le \tilde u_k$, we have
\begin{equation}
\bar v_k(x) 
= \tilde v_k(x) -\eta_k 
\le \tilde u_k(x) -\eta_k 
\le u(x), \quad \forall x \in [-a,a].
\end{equation}
In particular, $\bar h_k := \exp(\bar v_k)$ is feasible for \eqref{logconcave_original_maximization}, implying that $\int \bar h_k \le \pi_0$.
On the other hand, we have
\begin{equation}
\int \bar h_k
= \int \exp(\bar v_k)
= \exp(-\eta_k) \int_{[-a,a]} \exp(\tilde v_k)
= \exp(-\eta_k) \int_{[-a,a]} \tilde h_k,
\end{equation}
and so, because $\tilde h_k \le \tilde f_k := \exp(\tilde u_k)$,
\begin{align*}
\pi_{0,k}
= \Lambda(\bv_k)
= \int \tilde h_k
&= \int_{[-a,a]} \tilde h_k + \int_{[-a,a]^\comp} \tilde h_k \\
&= \exp(\eta_k) \int \bar h_k + \int_{[-a,a]^\comp} \tilde h_k \\
&\le \exp(\eta_k) \pi_0 + \int_{[-a,a]^\comp} \tilde f_k.
\end{align*}
Given $\eps > 0$ arbitrarily small, choose $a$ large enough that $\int_{[-a,a]^\comp} f \le \eps$. Since 
\begin{equation}
\int_{[-a,a]^\comp} \tilde f_k \xrightarrow{k \to\infty} \int_{[-a,a]^\comp} f,
\end{equation}
for $k$ large enough we have $\int_{[-a,a]^\comp} \tilde f_k \le 2\eps$, implying that $\pi_{0,k} \le \exp(\eta_k) \pi_0 + 2\eps$. Using the fact that $\exp(\eta_k) \to 1$ since $\eta_k = L \delta_k \to 0$, and that $\eps > 0$ is arbitrary, we get that
\begin{equation}
\limsup_{k\to\infty} \pi_{0,k}
\le \pi_0,
\end{equation}
concluding the proof.
\end{proof}

We mention that besides the discretization \eqref{actual_linear_approximation}, we also considered a more straightforward discretization of the integral, effectively replacing $\Lambda$ in \eqref{actual_linear_approximation} with $\Lambda_0$ defined as
\begin{equation}
    \label{riemann_linear_approximation}
\Lambda_0(v_{-k}, \dots, v_k) := \delta_k \left(\tfrac{1}{2} \exp(v_{-k}) + \tfrac{1}{2}\exp(v_k) + \sum_{j=-k+1}^{k-1} \exp(v_j)\right).
\end{equation}
The outputs returned by these two discretizations, \eqref{actual_linear_approximation} and \eqref{riemann_linear_approximation}, were very similar in our numerical experiments.

\subsection{Numerical experiments}
\label{sec:logconcave_simulation_subsection}

In this subsection we consider experiments where the background component could be assumed to be log-concave. The density fitting process and confidence band acquisition are exactly the same as in \secref{symmetric_experiments}. We first consider four different situations as listed in \tabref{logcon_simulation_situations}. Although the mixture distributions are identical to those in \tabref{symmetric_simulation_situations} for the symmetric case, we need to point out that $\pi_0$ is different here, as $\pi_0$ here corresponds to the largest possible log-concave component as defined in \eqref{logconcave_pi0_definition}. We again generate a sample of size $n = 1000$ from each model, and repeat each setting $1000$ times. 

We note that the output of our algorithm depends heavily on the estimation of the density $\hat{f}$. When the bandwidth selected by cross-validation yields a density with a high frequency of oscillation, the largest log-concave component is very likely to be smaller than the correct value. For an illustrative situation, see \figref{illustration_logconcave_vascillation}. In the event that such a situation happen, we recommend that the user look at a plot of $\hat{f}$ before applying our procedure, with the possibility of selecting a larger bandwidth. This is what we did, for example, for the Carina and Leukemia datasets in \figref{logconcave_bandwidth_adjustment}. In the simulations, to avoid the effect of these issue, we report the median value instead of the mean. These values are reported in  \tabref{logcon_simulation_numbers}. 

As can be seen, even with only the assumption of log-concavity, our method is accurate in estimating $\pi_0$, with estimation errors ranging from $0.001$ to $0.007$. Our method frequently achieves smaller error in estimating $\pi_0$ than comparison methods in estimating $\theta_0$ and $\theta$, and does so with less information on the background component. For situations where the background component is specified incorrectly, we invite the reader to \secref{incorrect_background_subsection}.

\begin{figure}[htpb]
    \centering
    \includegraphics[scale=0.35]{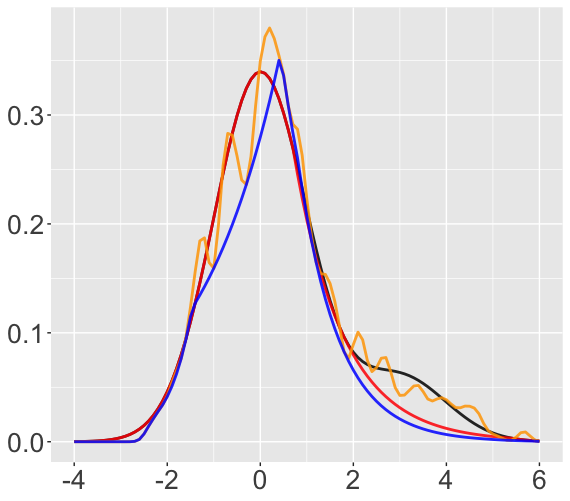}
    \caption{An illustrative situation that occurred in the course of our simulations, where a high frequency of oscillation of $\hat{f}$ results in a significantly lower estimation for $\pi_0$: a Gaussian mixture density $0.85 \text{ } \mathcal{N}(0, 1) + 0.15 \text{ } \mathcal{N} (0, 1)$, in black, and corresponding largest log-concave component
    $h$, in red, with the fitted density $\hat{f}$, in orange, and fitted largest log-concave component $\hat{h}$, in blue. Here, as before, $\hat{f}$ is obtained by kernel density estimation with bandwidth chosen by cross-validation. The result based on this estimate is $\hat{\pi}_0 = 0.807$, while the actual value is $\pi_0 = 0.931$.}
    \label{fig:illustration_logconcave_vascillation}
\end{figure}

\begin{rem}
We note that here the SQP algorithm as implemented in the R package \textsf{nloptr} sometimes gives $\hat{h} = 0$ as the largest log-concave component due to some parts of the estimated density $\hat{f}$ being $0$. This situation happens occasionally when computing the lower confidence bound, and it is of course incorrect in situations like those in \tabref{logcon_simulation_situations}. When this situation occurs, we rerun the algorithm on the largest interval of non-zero $\hat{f}$ values and report the greatest log-concave component acquired on that interval.
\end{rem}

\begin{table}[htpb]
\centering\small
\caption{Log-concave simulation situations of Gaussian mixtures, as well as values of $\theta$, $\theta_0$, and $\pi_0$, obtained through numerical optimization.}
\label{tab:logcon_simulation_situations}
\bigskip
\setlength{\tabcolsep}{0.03in}
\begin{tabular}{p{0.1\textwidth} 
p{0.55\textwidth} p{0.1\textwidth} p{0.1\textwidth} p{0.1\textwidth} 
}
\toprule
{\bf Model} & {\bf Distribution} & {\bf $\theta $} & {\bf $\theta_0 $} & {\bf $\pi_0$}\\ 
\midrule
\textsf{L1} & $0.85 \text{ }\mathcal{N} (0, 1) + 0.15 \text{ }\mathcal{N}(3, 1) $ & 0.850 & 0.850 & 0.931 \\ 
\midrule
\textsf{L2} & $0.95 \text{ }\mathcal{N} (0, 1) + 0.05 \text{ }\mathcal{N}(3, 1) $ & 0.950 & 0.950 & 0.981 \\ 
\midrule
\textsf{L3} & $0.85 \text{ }\mathcal{N} (0, 1) + 0.1 \text{ }\mathcal{N}(2.5, 0.75) + 0.05 \text{ }\mathcal{N}(-2.5, 0.75) $ & 0.850 & 0.851 &  0.975\\
\midrule
\textsf{L4} & $0.85 \text{ }\mathcal{N} (0, 1) + 0.1 \text{ }\mathcal{N}(2.5, 0.75) + 0.05 \text{ }\mathcal{N}(5, 0.75) $ & 0.850 & 0.850 & 0.946\\

\bottomrule
\end{tabular}
\end{table}

\begin{table}[htpb]
\centering\small
\caption{A comparison of various methods for estimating a log-concave background component in the situations of \tabref{logcon_simulation_situations}. Here we report the median values instead of the mean values (and also report the standard deviations, as before). As always, $\hat{\pi}_0$ should be compared with $\pi_0$, $\hat{\theta}_0^X$ should be compared with $\theta_0$, and $\hat{\theta}^X$ should be compared with $\theta$.}
\label{tab:logcon_simulation_numbers}
\bigskip
\setlength{\tabcolsep}{0.025in}
\begin{tabular}{p{0.1\textwidth} 
p{0.075\textwidth} p{0.075\textwidth} p{0.075\textwidth} 
p{0.1\textwidth}
p{0.075\textwidth} p{0.075\textwidth} p{0.075\textwidth} p{0.075\textwidth} 
p{0.075\textwidth} p{0.075\textwidth}
}
\toprule
{\bf Model} & {\bf $\hat{\pi}_0$} & {\bf $\hat{\pi}_0^{\rm{L}}$} & {\bf $\hat{\pi}_0^{\rm{U}}$} & {\bf Null} & {\bf $\hat{\theta}_0^{\rm{PSC}}$} & {\bf $\hat{\theta}_0^{\rm{PSH}}$} & {\bf $\hat{\theta}_0^{\rm{PSB}}$} & {\bf $\hat{\theta}^{\rm{E}}$} & {\bf $\hat{\theta}^{\rm{MR}}$} & {\bf $\hat{\theta}^{\rm{CJ}}$}
\\ 
\midrule
\textsf{L1}   & 0.932 & 0.596 & 1 & $\mathcal{N} (0,1)$ & 0.850 & 0.858 & 0.894 & 0.861 & 0.890 & 0.888 \\ 
   & (0.034) & (0.074) & (0) & & (0.024) & (0.023) & (0.018) & (0.023) & (0.014) & (0.085) \\ 
\midrule
\textsf{L2}   & 0.974 & 0.662 & 1 & $\mathcal{N}(0,1)$ & 0.942 & 0.958 & 0.988 & 0.953 & 0.972 & 0.964 \\ 
\text{   }   & (0.028) & (0.080) & (0) & & (0.023) & (0.022) & (0.015) & (0.026) & (0.008) & (0.082) \\
\midrule
\textsf{L3} & 0.969 & 0.611 & 1 & $\mathcal{N} (0,1)$ & 0.864 & 0.948 & 0.936 & 0.856 & 0.896 & 0.705
\\
 & (0.031) & (0.077) & (0) & & (0.023) & (0.034) & (0.017) & (0.024) & (0.015) & (0.085)
\\
\midrule
\textsf{L4} & 0.942 & 0.632 & 1 & $\mathcal{N}(0,1)$ & 0.848 & 0.857 & 0.893 & 0.849 & 0.890 & 0.712
\\ 
\text{   }   & (0.033) & (0.074) & (0) & & (0.023) & (0.024) & (0.018) & (0.024) & (0.014) & (0.088) \\

\bottomrule
\end{tabular}
\end{table}

\subsection{Real data analysis}

In this subsection we examine real datasets of two component mixtures where the background component could be assumed to be log-concave. We first consider the six real datasets presented in \secref{symmetric_read_data}, but this time look for a background log-concave component. When the background is known, the numerical results of the Prostate and Carina datasets can be found in \tabref{logconcave_realdata_known_background}, \figref{logconcave_prostate_picture} and \figref{logconcave_carina_picture}. When the background is unknown, the numerical results of the HIV, Leukemia, Parkinson and Police datasets can be found in \tabref{logconcave_realdata_unknown_background}, \figref{logconcave_hiv_picture},  \figref{logconcave_leukemia_picture}, \figref{logconcave_parkinson_picture} and \figref{logconcave_police_picture}.

In addition to the above six datasets, we also include here the Old Faithful Geyser dataset \citep{azzalini1990look}. This dataset consists of $272$ waiting times, in minutes, between eruptions for the Old Faithful Geyser in Yellowstone National Park, Wyoming, USA. We attempt to find the largest log-concave component of the waiting times. The results are summarized in the first row of \tabref{logconcave_realdata_unknown_background}, \figref{logconcave_geyser_with_band} and \figref{logconcave_geyser_without_band}. We want to note from this example that the curve $\hat{h}_u$ acquired from $\hat{f}_u$ leading to the computation of $\hat{\pi}_0^{\rm U}$, is not the upper confidence bound of $h_0$, as shown by \figref{logconcave_geyser_with_band}. 

\begin{rem}
We observe here through these real datasets that compared to symmetric background assumptions in \secref{symmetric_read_data}, the largest log-concave background component usually has a higher weight than the largest symmetric background component.
\end{rem}

\begin{table}[htpb]
\centering\small
\caption{Real datasets where the background log-concave component is known. (Note that we work with $z$ values here instead of $p$-values so our result in Prostate dataset is slightly different from that in \citep{patra}.)}
\label{tab:logconcave_realdata_known_background}
\bigskip
\setlength{\tabcolsep}{0.03in}
\begin{tabular}{
p{0.1\textwidth} 
p{0.07\textwidth} p{0.07\textwidth} p{0.05\textwidth} 
p{0.08\textwidth}
p{0.07\textwidth} p{0.07\textwidth} p{0.07\textwidth} p{0.07\textwidth}
p{0.07\textwidth} p{0.07\textwidth}
}
\toprule
{\bf Model} & {\bf $\hat{\pi}_0$} & {\bf $\hat{\pi}_0^{\rm{L}}$} & {\bf $\hat{\pi}_0^{\rm{U}}$} &  {\bf Null} & {\bf $\hat{\theta}_0^{\rm{PSC}}$} & {\bf $\hat{\theta}_0^{\rm{PSH}}$} & {\bf $\hat{\theta}_0^{\rm{PSB}}$} & {\bf $\hat{\theta}^{\rm{E}}$} & {\bf $\hat{\theta}^{\rm{MR}}$} & {\bf $\hat{\theta}^{\rm{CJ}}$} \\ 
\midrule
Prostate & 0.994 & 0.809 & 1 & $\mathcal{N} (0,1)$ & 0.931 & 0.941 & 0.975 & 0.931 & 0.956 & 0.867  \\
\midrule
Carina & 0.600 & 0.242 & 1 & \textsf{bgstars} & 0.636 & 0.645 & 0.677 & 0.951 & 0.664 & 0.206  \\
\bottomrule
\end{tabular}
\end{table}

\begin{table}[htpb]
\centering\small
\caption{Real datasets where the background log-concave distribution is unknown. (Note that Efron's method will not run on the Geyser dataset.)}
\label{tab:logconcave_realdata_unknown_background}
\bigskip
\setlength{\tabcolsep}{0.05in}
\begin{tabular}{
p{0.15\textwidth} 
p{0.1\textwidth} p{0.1\textwidth} p{0.1\textwidth} 
p{0.1\textwidth} p{0.1\textwidth} 
}
\toprule
{\bf Model} & {\bf $\hat{\pi}_0$} & {\bf $\hat{\pi}_0^{\rm{L}}$} & {\bf $\hat{\pi}_0^{\rm{U}}$} &  {\bf $\hat{\theta}^{\rm{E}}$ } & {\bf $\hat{\theta}^{\rm{CJ}}$ }\\ 
\midrule
Geyser & 0.693 & 0.287 & 1 & NA & 1  \\
\midrule
HIV & 0.984 & 0.804 & 1 & 0.940 & 0.926 \\
\midrule
Leukemia & 0.981 & 0.695 & 1 & 0.911 & 0.820 \\
\midrule
Parkinson & 1 & 0.605 & 1 & 0.998 & 0.993 \\
\midrule
Police & 0.997 & 0.765 & 1 & 0.985  & 0.978 \\

\bottomrule
\end{tabular}
\end{table}

\begin{figure}[htpb]

    \centering
    \centering
    \subfigure[Prostate dataset]{\label{fig:logconcave_prostate_picture}\includegraphics[scale=0.35]{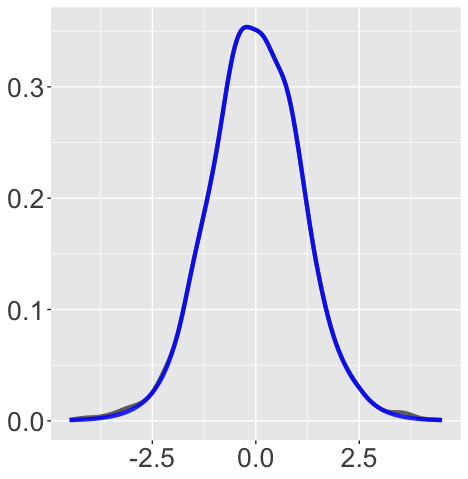}} \qquad
    \centering
    \subfigure[Carina dataset]{\label{fig:logconcave_carina_picture}\includegraphics[scale=0.35]{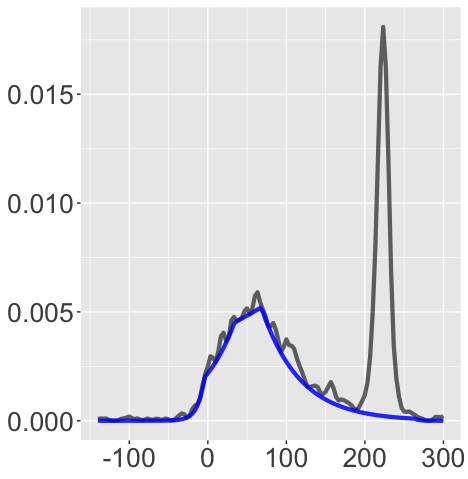}}

    \caption{Estimated background log-concave component on the Prostate ($z$ values) and Carina (radial velocity) datasets: the black curve represents fitted density; the blue curve represents computed $\hat{h}$, one of the largest log-concave components. Due to unusually high frequency of oscillation in $\hat{f}$ in the Carina dataset, we consider increasing the bandwidth from the one acquired by cross-validation, and the corresponding result is shown in \figref{logconcave_carina_adj_bandwidth}.}
\end{figure}

\begin{figure}[htpb]

    \centering

    \centering
    \subfigure[HIV dataset]{\label{fig:logconcave_hiv_picture}\includegraphics[scale=0.35]{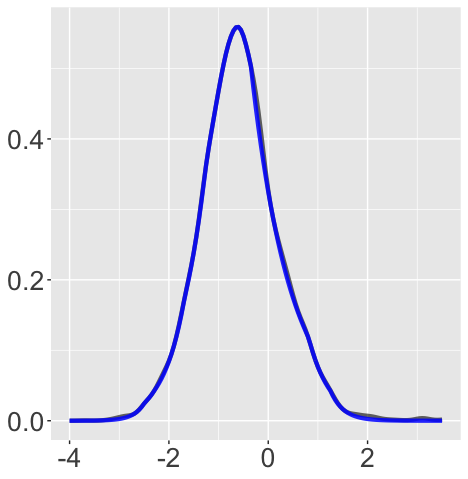}} \qquad
    \centering
    \subfigure[Leukemia dataset]{\label{fig:logconcave_leukemia_picture}\includegraphics[scale=0.35]{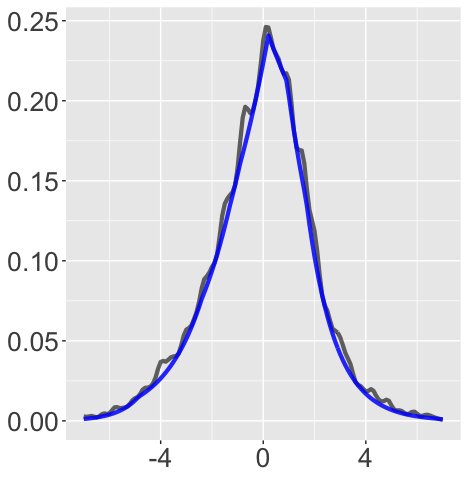}} \\
    \centering
    \subfigure[Parkinson dataset]{\label{fig:logconcave_parkinson_picture}\includegraphics[scale=0.35]{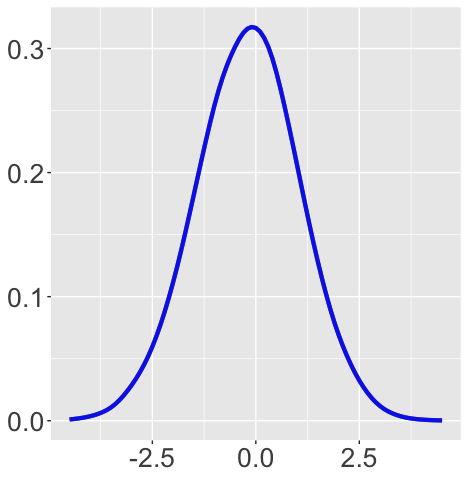}} \qquad
    \centering
    \subfigure[Police dataset]{\label{fig:logconcave_police_picture}\includegraphics[scale=0.35]{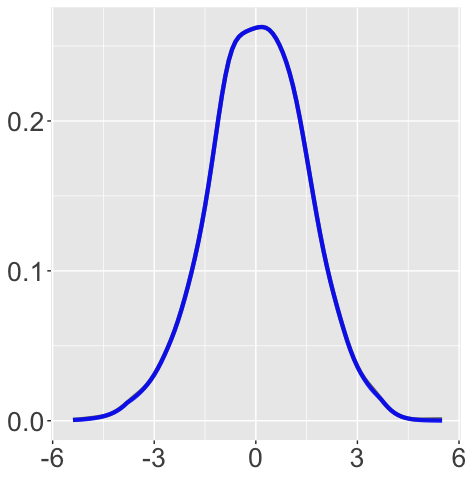}}
    
    \caption{Estimated background log-concave component on HIV ($z$ values), Leukemia ($z$ values), Parkinson ($z$ values), and Police ($z$ scores) datasets: the black curve represents fitted density; the blue curve represents computed $\hat{h}$, one of the largest log-concave components. Due to the high frequency of oscillation in $\hat{f}$ in the Leukemia dataset, we consider increasing the bandwidth from the one acquired by cross-validation, and the corresponding result is shown in \figref{logconcave_leukemia_adj_bandwidth}. }
    \label{fig:logconcave_other_real_data}
\end{figure}

\begin{figure}[htpb]

    \centering
    \centering
    \subfigure[Carina dataset]
    {
    \label{fig:logconcave_carina_adj_bandwidth}
    \includegraphics[scale=0.35]{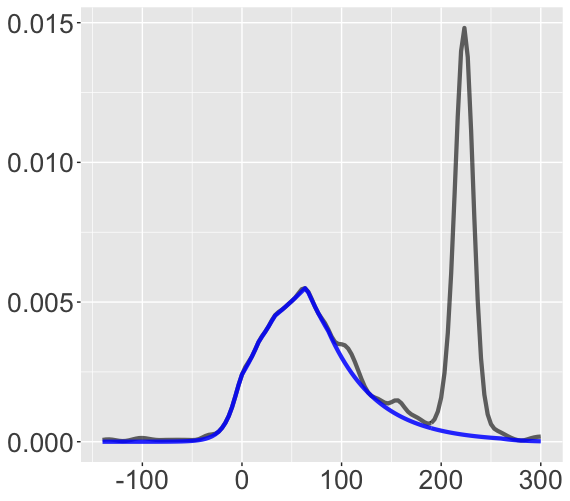}} \qquad
    \centering
    \subfigure[Leukemia dataset]{
    \label{fig:logconcave_leukemia_adj_bandwidth}
    \includegraphics[scale=0.35]{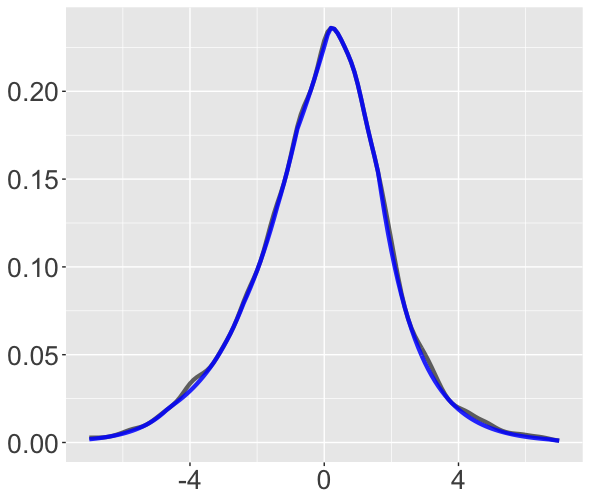}}

    \caption{Carina (radial velocity) and Leukemia ($z$ values) datasets with kernel density, with bandwidth chosen `by hand' instead of by cross-validation: the black curve represents fitted density with increased bandwidth; the blue curve represents the computed $\hat{h}$. We note that for the Carina dataset, the bandwidth was increased from $3.085$ to $6$, and $\hat{\pi}_0$ changed from $0.550$ to $0.600$. For the Leukemia dataset, the bandwidth was been increased from $0.124$ to $0.25$, and $\hat{\pi}_0$ changed from $ 0.939$ to $0.981$.}
    \label{fig:logconcave_bandwidth_adjustment}
\end{figure}

\begin{figure}[htpb]
    \centering
    
    \centering
    \subfigure[Geyser dataset with $\hat{h}_l$ and $\hat{h}_u$]{\label{fig:logconcave_geyser_with_band}\includegraphics[scale=0.35]{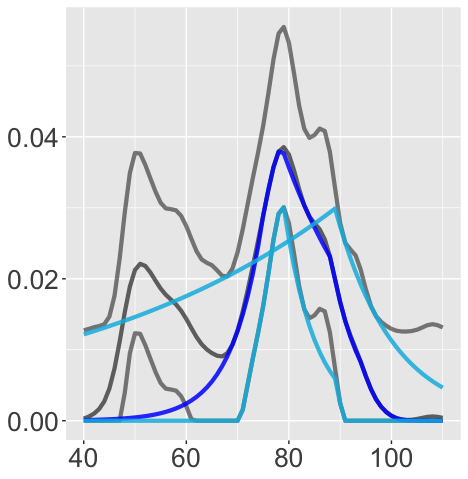}} \qquad
    \centering
    \subfigure[Geyser dataset with only $\hat{h}$ ]{\label{fig:logconcave_geyser_without_band}\includegraphics[scale=0.35]{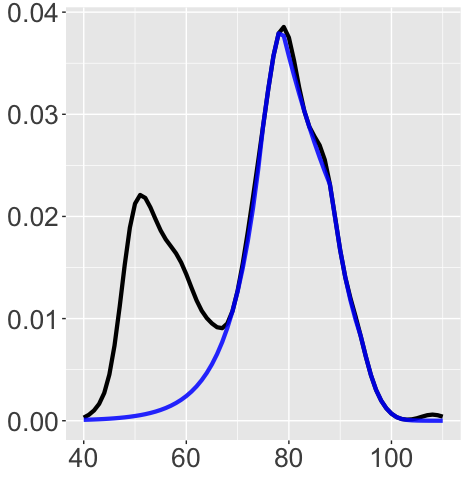}}
    \caption{Estimated background log-concave component on the Geyser (duration) dataset: the black curve represents the fitted density; the blue curve represents the computed $\hat{h}$; the gray curves (left) represent $\hat{f}_l$ and $\hat{f}_u$; the light blue curves (left) represents $\hat{h}_l$ and $\hat{h}_u$. Note from the left plot that $\hat{h}_u$ is only used to compute $\hat{\pi}_0^{\rm U}$, and is not an upper bound for $h$. }
\end{figure}

\section{Conclusion and discussion}
\label{sec:discussion}

In this paper, we extend the approach of \cite{patra} to settings where the background component of interest is assumed to belong to three emblematic examples: symmetric, monotone, and log-concave. In each setting, we derive estimators for both the proportion and density for the background component, establish their consistency, and provide confidence intervals/bands. However, the important situation of incorrect background distribution and other extensions remain unaddressed, and therefore we discuss them in this section.

\subsection{Incorrect background specification}
\label{sec:incorrect_background_subsection}

As mentioned in \secref{symmetric_experiments} and \secref{logconcave_simulation_subsection}, our method requires much less information than comparable methods, and therefore is much less prone to misspecification of the background component. In this subsection, we give an experiment illustrating this point. We consider the mixture model: 
$
0.85 \text{ }\mathcal{T}_6 + 0.15 \text{ }\mathcal{N}(3, 1).
$
Instead of the correct null distribution $\mathcal{T}_6$, we take it to be $\mathcal{N} (0,1)$. This situation could happen in situations of multiple testing settings, for example in the HIV dataset in \secref{symmetric_read_data}. We consider both a symmetric background and log-concave background on this model in \tabref{incorrect_background_situations}, and report the fitted values of our method and comparison methods in \tabref{incorrect_background_results_table}. As can be seen, when estimating $\pi_0$, our estimator achieves $0.002$ error when assuming symmetric background and $0.004$ error when assuming log-concave background, less then any other method in comparison that estimates $\theta_0$ or $\theta$. The heuristic estimator of \cite{patra} has slightly higher error than our method, while the constant estimator of \cite{patra} and the estimators of \cite{efron2007size} and \cite{cai2010optimal} have large errors. The upper confidence bound of \cite{meinshausen2006estimating} also becomes incorrect.

\begin{table}[htpb]
\centering\small
\caption{Simulation situations when background specification is incorrect, as well as values of $\theta$, $\theta_0$, $\pi_0$, obtained through numerical optimization.}
\label{tab:incorrect_background_situations}
\bigskip
\setlength{\tabcolsep}{0.05in}
\begin{tabular}{p{0.1\textwidth} p{0.2\textwidth}
p{0.3\textwidth} p{0.1\textwidth} p{0.1\textwidth} p{0.1\textwidth} 
}
\toprule
{\bf Model} & {\bf Background} & {\bf Distribution} & {\bf $\theta $} & {\bf $\theta_0$} & {\bf $\pi_0$} \\ 
\midrule
\textsf{S5} & Symmetric & $0.85 \text{ }\mathcal{T}_6 + 0.15 \text{ }\mathcal{N}(3, 1) $ & 0.850 & 0.850 & 0.859 \\ 
\midrule
\textsf{L5} & Log-concave & $0.85 \text{ }\mathcal{T}_6 + 0.15 \text{ }\mathcal{N}(3, 1) $ & 0.850 & 0.850 & 0.925 \\ 
\bottomrule
\end{tabular}
\end{table}

\begin{table}[htpb]
\centering\small
\caption{Results for the situations of \tabref{incorrect_background_situations}. We provide mean values in \textsf{S5}, and median values in \textsf{L5}.}
\label{tab:incorrect_background_results_table}

\bigskip
\setlength{\tabcolsep}{0.02in}
\begin{tabular}{p{0.08\textwidth} 
p{0.075\textwidth} 
p{0.075\textwidth} p{0.075\textwidth} p{0.075\textwidth} 
p{0.08\textwidth}
p{0.075\textwidth} p{0.075\textwidth} p{0.075\textwidth} p{0.075\textwidth} p{0.075\textwidth} p{0.075\textwidth}}
\toprule
{\bf Model} & {Center} & {\bf $\hat{\pi}_0$} & {\bf $\hat{\pi}_0^{\rm{L}}$} & {\bf $\hat{\pi}_0^{\rm{U}}$} & {\bf Null} & {\bf $\hat{\theta}_0^{\rm{PSC}}$} & {\bf $\hat{\theta}_0^{\rm{PSH}}$} & {\bf $\hat{\theta}_0^{\rm{PSB}}$} & {\bf $\hat{\theta}^{\rm{E}}$} & {\bf $\hat{\theta}^{\rm{MR}}$} & {\bf $\hat{\theta}^{\rm{CJ}}$}\\ 
\midrule
\textsf{S5} & 0.073 & 0.857 & 0.541 & 1 & $\mathcal{N} (0,1)$ & 0.815 & 0.855 & 0.877 & 0.803 & 0.843 & 0.816 \\ 
\text{   } & (0.057) & (0.021) & (0.057) & (0) & & (0.021) & (0.022) & (0.016) & (0.026) & (0.015) & (0.086)\\ 
\midrule
\textsf{L5} & & 0.921 & 0.574 & 1 & $\mathcal{N} (0,1)$ & 0.817 & 0.856 & 0.877 & 0.803 & 0.843 & 0.814  \\ 
\text{   } &  & (0.036) & (0.069) & (0) & & (0.021) & (0.022) & (0.016) & (0.026) & (0.015) & (0.086)\\ 
\bottomrule
\end{tabular}
\end{table}

\subsection{Combinations}

Although not discussed in the main part of the paper, some combinations of the shape constraints considered earlier are possible. For example, one could consider extracting a maximal background that is symmetric {\em and} log-concave; or one could consider extracting a maximal background that is monotone {\em and} log-concave. As it turns out, these two combinations are intimately related. Mixtures of symmetric log-concave distributions are considered, for example, in \citep{pu2020algorithm}.

\subsection{Generalization to higher dimensions}

All our examples were on the real line, corresponding to real-valued observations, simply because the work was in large part motivated by multiple testing in which the sample stands for the test statistics. But the approach is more general. Indeed, consider a measurable space, and let $\cD$ be a class of probability distributions on that space. Given a probability distribution $\mu$, we can quantify how much there is of $\cD$ in $\mu$ by defining
\begin{equation}
\pi_0 := \sup \big\{\pi: \exists \nu \in \cD \text{ s.t. } \mu \ge  \pi \nu\big\}.
\end{equation}

For concreteness, we give a simple example in an arbitrary dimension $d$ by generalizing the setting of a symmetric background component covered in \secref{symmetric}. 
Although various generalizations are possible, we consider the class --- also denoted $\cS$ as in \eqref{symmetric_pi0_definition} --- of spherically symmetric (i.e., radial) densities with respect to the Lebesgue measure on $\bbR^d$.
It is easy to see that the background component proportion is given by 
\begin{align}
\pi_0 = \int_{\bbR^d} h_0(x) d x,
&& h_0(x) := \min \{ f(y) : \|y\| = \|x\|\},
\end{align}
and, if $\pi_0 > 0$, the background component density is given by $g_0 := h_0/\pi_0$.

\subsection*{Acknowledgments}
We are grateful to Philip Gill for discussions regarding the discretization of the optimization problem \eqref{logconcave_original_maximization}.  

\bibliographystyle{chicago}
\bibliography{two_component.bib}

\begin{thebibliography}{}

\bibitem[\protect\citeauthoryear{Arias-Castro and Chen}{Arias-Castro and
  Chen}{2017}]{arias2017distribution_fdr}
Arias-Castro, E. and S.~Chen (2017).
\newblock Distribution-free multiple testing.
\newblock {\em Electronic Journal of Statistics\/}~{\em 11\/}(1), 1983--2001.

\bibitem[\protect\citeauthoryear{Arias-Castro and Wang}{Arias-Castro and
  Wang}{2017}]{arias2017distribution}
Arias-Castro, E. and M.~Wang (2017).
\newblock Distribution-free tests for sparse heterogeneous mixtures.
\newblock {\em TEST\/}~{\em 26\/}(1), 71--94.

\bibitem[\protect\citeauthoryear{Arlot and Celisse}{Arlot and
  Celisse}{2010}]{arlot2010survey}
Arlot, S. and A.~Celisse (2010).
\newblock A survey of cross-validation procedures for model selection.
\newblock {\em Statistics Surveys\/}~{\em 4}, 40--79.

\bibitem[\protect\citeauthoryear{Azzalini and Bowman}{Azzalini and
  Bowman}{1990}]{azzalini1990look}
Azzalini, A. and A.~W. Bowman (1990).
\newblock A look at some data on the {Old Faithful Geyser}.
\newblock {\em Journal of the Royal Statistical Society: Series C\/}~{\em
  39\/}(3), 357--365.

\bibitem[\protect\citeauthoryear{Benjamini and Hochberg}{Benjamini and
  Hochberg}{1995}]{benjamini1995controlling}
Benjamini, Y. and Y.~Hochberg (1995).
\newblock Controlling the false discovery rate: a practical and powerful
  approach to multiple testing.
\newblock {\em Journal of the Royal Statistical Society: Series B\/}~{\em
  57\/}(1), 289--300.

\bibitem[\protect\citeauthoryear{Benjamini and Hochberg}{Benjamini and
  Hochberg}{2000}]{benjamini2000adaptive}
Benjamini, Y. and Y.~Hochberg (2000).
\newblock On the adaptive control of the false discovery rate in multiple
  testing with independent statistics.
\newblock {\em Journal of Educational and Behavioral Statistics\/}~{\em
  25\/}(1), 60--83.

\bibitem[\protect\citeauthoryear{Bickel and Rosenblatt}{Bickel and
  Rosenblatt}{1973}]{bickel1973some}
Bickel, P.~J. and M.~Rosenblatt (1973).
\newblock On some global measures of the deviations of density function
  estimates.
\newblock {\em The Annals of Statistics\/}~{\em 1\/}(6), 1071--1095.

\bibitem[\protect\citeauthoryear{Bordes, Mottelet, and Vandekerkhove}{Bordes
  et~al.}{2006}]{bordes2006semiparametric}
Bordes, L., S.~Mottelet, and P.~Vandekerkhove (2006).
\newblock Semiparametric estimation of a two-component mixture model.
\newblock {\em The Annals of Statistics\/}~{\em 34\/}(3), 1204--1232.

\bibitem[\protect\citeauthoryear{Cai and Jin}{Cai and
  Jin}{2010}]{cai2010optimal}
Cai, T.~T. and J.~Jin (2010).
\newblock Optimal rates of convergence for estimating the null density and
  proportion of nonnull effects in large-scale multiple testing.
\newblock {\em The Annals of Statistics\/}~{\em 38\/}(1), 100--145.

\bibitem[\protect\citeauthoryear{Cattaneo, Jansson, and Ma}{Cattaneo
  et~al.}{2019}]{cattaneo2019lpdensity}
Cattaneo, M.~D., M.~Jansson, and X.~Ma (2019).
\newblock {\sf lpdensity}: Local polynomial density estimation and inference.
\newblock {\em arXiv preprint arXiv:1906.06529\/}.

\bibitem[\protect\citeauthoryear{Cattaneo, Jansson, and Ma}{Cattaneo
  et~al.}{2020}]{cattaneo2020simple}
Cattaneo, M.~D., M.~Jansson, and X.~Ma (2020).
\newblock Simple local polynomial density estimators.
\newblock {\em Journal of the American Statistical Association\/}~{\em
  115\/}(531), 1449--1455.

\bibitem[\protect\citeauthoryear{Chang and Walther}{Chang and
  Walther}{2007}]{chang2007clustering}
Chang, G.~T. and G.~Walther (2007).
\newblock Clustering with mixtures of log-concave distributions.
\newblock {\em Computational Statistics \& Data Analysis\/}~{\em 51\/}(12),
  6242--6251.

\bibitem[\protect\citeauthoryear{Chen}{Chen}{2017}]{chen2017tutorial}
Chen, Y.-C. (2017).
\newblock A tutorial on kernel density estimation and recent advances.
\newblock {\em Biostatistics \& Epidemiology\/}~{\em 1\/}(1), 161--187.

\bibitem[\protect\citeauthoryear{Cheng and Chen}{Cheng and
  Chen}{2019}]{cheng2019nonparametric}
Cheng, G. and Y.-C. Chen (2019).
\newblock Nonparametric inference via bootstrapping the debiased estimator.
\newblock {\em Electronic Journal of Statistics\/}~{\em 13\/}(1), 2194--2256.

\bibitem[\protect\citeauthoryear{Chow, Geman, and Wu}{Chow
  et~al.}{1983}]{chow1983consistent}
Chow, Y.-S., S.~Geman, and L.-D. Wu (1983).
\newblock Consistent cross-validated density estimation.
\newblock {\em The Annals of Statistics\/}~{\em 11\/}(1), 25--38.

\bibitem[\protect\citeauthoryear{Cline and Hart}{Cline and
  Hart}{1991}]{cline1991kernel}
Cline, D. B.~H. and J.~D. Hart (1991).
\newblock Kernel estimation of densities with discontinuities or discontinuous
  derivatives.
\newblock {\em Statistics\/}~{\em 22\/}(1), 69--84.

\bibitem[\protect\citeauthoryear{Cohen}{Cohen}{1967}]{cohen1967estimation}
Cohen, A.~C. (1967).
\newblock Estimation in mixtures of two normal distributions.
\newblock {\em Technometrics\/}~{\em 9\/}(1), 15--28.

\bibitem[\protect\citeauthoryear{Dong, Du, and Gardner}{Dong
  et~al.}{2021}]{coronavirus}
Dong, E., H.~Du, and L.~Gardner (2021).
\newblock An interactive web-based dashboard to track {COVID}-19 in real time.
\newblock {\em The Lancet: Infectious Diseases\/}~{\em 20\/}(5), 533--534.

\bibitem[\protect\citeauthoryear{Efron}{Efron}{2007}]{efron2007size}
Efron, B. (2007).
\newblock Size, power and false discovery rates.
\newblock {\em The Annals of Statistics\/}~{\em 35\/}(4), 1351--1377.

\bibitem[\protect\citeauthoryear{Efron}{Efron}{2012}]{efron2012large}
Efron, B. (2012).
\newblock {\em Large-Scale Inference: Empirical Bayes Methods for Estimation,
  Testing, and Prediction}.
\newblock Cambridge University Press.

\bibitem[\protect\citeauthoryear{Efron, Tibshirani, Storey, and Tusher}{Efron
  et~al.}{2001}]{efron2001empirical}
Efron, B., R.~Tibshirani, J.~D. Storey, and V.~Tusher (2001).
\newblock Empirical {Bayes} analysis of a microarray experiment.
\newblock {\em Journal of the American Statistical Association\/}~{\em
  96\/}(456), 1151--1160.

\bibitem[\protect\citeauthoryear{Fan}{Fan}{1993}]{fan1993local}
Fan, J. (1993).
\newblock Local linear regression smoothers and their minimax efficiencies.
\newblock {\em The Annals of Statistics\/}~{\em 21\/}(1), 196--216.

\bibitem[\protect\citeauthoryear{Franc\r{u}, Kerman, and Sinnamon}{Franc\r{u}
  et~al.}{2017}]{francuu2017new}
Franc\r{u}, M., R.~Kerman, and G.~Sinnamon (2017).
\newblock A new algorithm for approximating the least concave majorant.
\newblock {\em Czechoslovak Mathematical Journal\/}~{\em 67\/}(4), 1071--1093.

\bibitem[\protect\citeauthoryear{Gadat, Kahn, Marteau, and
  Maugis-Rabusseau}{Gadat et~al.}{2020}]{gadat2020parameter}
Gadat, S., J.~Kahn, C.~Marteau, and C.~Maugis-Rabusseau (2020).
\newblock Parameter recovery in two-component contamination mixtures: The
  {$L_2$} strategy.
\newblock {\em Annales de l'Institut Henri Poincar{\'e}: Probabilit{\'e}s et
  Statistiques\/}~{\em 56\/}(2), 1391--1418.

\bibitem[\protect\citeauthoryear{Genovese and Wasserman}{Genovese and
  Wasserman}{2002}]{genovese2002operating}
Genovese, C. and L.~Wasserman (2002).
\newblock Operating characteristics and extensions of the false discovery rate
  procedure.
\newblock {\em Journal of the Royal Statistical Society: Series B\/}~{\em
  64\/}(3), 499--517.

\bibitem[\protect\citeauthoryear{Genovese and Wasserman}{Genovese and
  Wasserman}{2004}]{genovese2004stochastic}
Genovese, C. and L.~Wasserman (2004).
\newblock A stochastic process approach to false discovery control.
\newblock {\em The Annals of Statistics\/}~{\em 32\/}(3), 1035--1061.

\bibitem[\protect\citeauthoryear{Gill and Wong}{Gill and
  Wong}{2012}]{gill2012sequential}
Gill, P.~E. and E.~Wong (2012).
\newblock Sequential quadratic programming methods.
\newblock In J.~Lee and S.~Leyffer (Eds.), {\em Mixed Integer Nonlinear
  Programming}, pp.\  147--224. Springer.

\bibitem[\protect\citeauthoryear{Gin{\'e} and Nickl}{Gin{\'e} and
  Nickl}{2010}]{gine2010confidence}
Gin{\'e}, E. and R.~Nickl (2010).
\newblock Confidence bands in density estimation.
\newblock {\em The Annals of Statistics\/}~{\em 38\/}(2), 1122--1170.

\bibitem[\protect\citeauthoryear{Gin{\'e} and Nickl}{Gin{\'e} and
  Nickl}{2021}]{gine2021mathematical}
Gin{\'e}, E. and R.~Nickl (2021).
\newblock {\em Mathematical Foundations of Infinite-Dimensional Statistical
  Models}.
\newblock Cambridge University Press.

\bibitem[\protect\citeauthoryear{Golub, Slonim, Tamayo, Huard, Gaasenbeek,
  Mesirov, Coller, Loh, Downing, and Caligiuri}{Golub
  et~al.}{1999}]{golub1999molecular}
Golub, T.~R., D.~K. Slonim, P.~Tamayo, C.~Huard, M.~Gaasenbeek, J.~P. Mesirov,
  H.~Coller, M.~L. Loh, J.~R. Downing, and M.~A. Caligiuri (1999).
\newblock Molecular classification of cancer: class discovery and class
  prediction by gene expression monitoring.
\newblock {\em Science\/}~{\em 286\/}(5439), 531--537.

\bibitem[\protect\citeauthoryear{Gorokhovik}{Gorokhovik}{2019}]{gorokhovik2018minimal}
Gorokhovik, V.~V. (2019).
\newblock Minimal convex majorants of functions and {Demyanov–Rubinov}
  exhaustive super(sub)differentials.
\newblock {\em Optimization\/}~{\em 68\/}(10), 1933--1961.

\bibitem[\protect\citeauthoryear{Gu}{Gu}{1993}]{gu1993smoothingalgorithm}
Gu, C. (1993).
\newblock Smoothing spline density estimation: a dimensionless automatic
  algorithm.
\newblock {\em Journal of the American Statistical Association\/}~{\em
  88\/}(422), 495--504.

\bibitem[\protect\citeauthoryear{Gu and Qiu}{Gu and
  Qiu}{1993}]{gu1993smoothingtheory}
Gu, C. and C.~Qiu (1993).
\newblock Smoothing spline density estimation: theory.
\newblock {\em The Annals of Statistics\/}~{\em 21\/}(1), 217--234.

\bibitem[\protect\citeauthoryear{Guidoum}{Guidoum}{2020}]{guidoum2015kernel}
Guidoum, A.~C. (2020).
\newblock Kernel estimator and bandwidth selection for density and its
  derivatives: The {\textsf{kedd}} package.
\newblock {\em arXiv preprint arXiv:2012.06102\/}.

\bibitem[\protect\citeauthoryear{Hettmansperger and McKean}{Hettmansperger and
  McKean}{2010}]{hettmansperger2010robust}
Hettmansperger, T.~P. and J.~W. McKean (2010).
\newblock {\em Robust Nonparametric Statistical Methods}.
\newblock CRC Press.

\bibitem[\protect\citeauthoryear{Hjort and Jones}{Hjort and
  Jones}{1996}]{hjort1996locally}
Hjort, N.~L. and M.~C. Jones (1996).
\newblock Locally parametric nonparametric density estimation.
\newblock {\em The Annals of Statistics\/}~{\em 24\/}(4), 1619--1647.

\bibitem[\protect\citeauthoryear{Hu, Wu, and Yao}{Hu
  et~al.}{2016}]{hu2016maximum}
Hu, H., Y.~Wu, and W.~Yao (2016).
\newblock Maximum likelihood estimation of the mixture of log-concave
  densities.
\newblock {\em Computational Statistics \& Data Analysis\/}~{\em 101},
  137--147.

\bibitem[\protect\citeauthoryear{Huber}{Huber}{1964}]{huber1964robust}
Huber, P.~J. (1964).
\newblock Robust estimation of a location parameter.
\newblock {\em The Annals of Mathematical Statistics\/}~{\em 35\/}(1), 73 --
  101.

\bibitem[\protect\citeauthoryear{Huber and Ronchetti}{Huber and
  Ronchetti}{2009}]{huber2009robust}
Huber, P.~J. and E.~M. Ronchetti (2009).
\newblock {\em Robust Statistics}.
\newblock Wiley.

\bibitem[\protect\citeauthoryear{Hunter, Wang, and Hettmansperger}{Hunter
  et~al.}{2007}]{hunter2007inference}
Hunter, D.~R., S.~Wang, and T.~P. Hettmansperger (2007).
\newblock Inference for mixtures of symmetric distributions.
\newblock {\em The Annals of Statistics\/}~{\em 35\/}(1), 224--251.

\bibitem[\protect\citeauthoryear{Jin}{Jin}{2008}]{jin2008proportion}
Jin, J. (2008).
\newblock Proportion of non-zero normal means: universal oracle equivalences
  and uniformly consistent estimators.
\newblock {\em Journal of the Royal Statistical Society: Series B\/}~{\em
  70\/}(3), 461--493.

\bibitem[\protect\citeauthoryear{Jin and Cai}{Jin and
  Cai}{2007}]{jin2007estimating}
Jin, J. and T.~T. Cai (2007).
\newblock Estimating the null and the proportion of nonnull effects in
  large-scale multiple comparisons.
\newblock {\em Journal of the American Statistical Association\/}~{\em
  102\/}(478), 495--506.

\bibitem[\protect\citeauthoryear{Jongbloed}{Jongbloed}{1998}]{jongbloed1998iterative}
Jongbloed, G. (1998).
\newblock The iterative convex minorant algorithm for nonparametric estimation.
\newblock {\em Journal of Computational and Graphical Statistics\/}~{\em
  7\/}(3), 310--321.

\bibitem[\protect\citeauthoryear{Karunamuni and Alberts}{Karunamuni and
  Alberts}{2005}]{karunamuni2005generalized}
Karunamuni, R.~J. and T.~Alberts (2005).
\newblock A generalized reflection method of boundary correction in kernel
  density estimation.
\newblock {\em Canadian Journal of Statistics\/}~{\em 33\/}(4), 497--509.

\bibitem[\protect\citeauthoryear{Kraft}{Kraft}{1988}]{kraft1988software}
Kraft, D. (1988).
\newblock A software package for sequential quadratic programming.
\newblock Technical Report DFVLR-FB 88-28, Oberpfaffenhofen: Institut für
  Dynamik der Flugsysteme.

\bibitem[\protect\citeauthoryear{Langaas, Lindqvist, and Ferkingstad}{Langaas
  et~al.}{2005}]{langaas2005estimating}
Langaas, M., B.~H. Lindqvist, and E.~Ferkingstad (2005).
\newblock Estimating the proportion of true null hypotheses, with application
  to {DNA} microarray data.
\newblock {\em Journal of the Royal Statistical Society: Series B\/}~{\em
  67\/}(4), 555--572.

\bibitem[\protect\citeauthoryear{Lesnick, Papapetropoulos, Mash,
  Ffrench-Mullen, Shehadeh, De~Andrade, Henley, Rocca, Ahlskog, and
  Maraganore}{Lesnick et~al.}{2007}]{lesnick2007genomic}
Lesnick, T.~G., S.~Papapetropoulos, D.~C. Mash, J.~Ffrench-Mullen, L.~Shehadeh,
  M.~De~Andrade, J.~R. Henley, W.~A. Rocca, J.~E. Ahlskog, and D.~M. Maraganore
  (2007).
\newblock A genomic pathway approach to a complex disease: {Axon} guidance and
  {Parkinson} disease.
\newblock {\em Public Library of Science Genetics\/}~{\em 3\/}(6), e98.

\bibitem[\protect\citeauthoryear{Lindsay}{Lindsay}{1995}]{lindsay1995mixture}
Lindsay, B.~G. (1995).
\newblock {\em Mixture models: theory, geometry and applications}, Volume~5 of
  {\em NSF-CBMS Regional Conference Series in Probability and Statistics}.
\newblock Institute of Mathematical Statistics.

\bibitem[\protect\citeauthoryear{Lindsay and Basak}{Lindsay and
  Basak}{1993}]{lindsay1993multivariate}
Lindsay, B.~G. and P.~Basak (1993).
\newblock Multivariate normal mixtures: a fast consistent method of moments.
\newblock {\em Journal of the American Statistical Association\/}~{\em
  88\/}(422), 468--476.

\bibitem[\protect\citeauthoryear{Loader}{Loader}{1996}]{loader1996local}
Loader, C.~R. (1996).
\newblock Local likelihood density estimation.
\newblock {\em The Annals of Statistics\/}~{\em 24\/}(4), 1602--1618.

\bibitem[\protect\citeauthoryear{Ma and Yao}{Ma and Yao}{2015}]{ma2015flexible}
Ma, Y. and W.~Yao (2015).
\newblock Flexible estimation of a semiparametric two-component mixture model
  with one parametric component.
\newblock {\em Electronic Journal of Statistics\/}~{\em 9\/}(1), 444--474.

\bibitem[\protect\citeauthoryear{McLachlan and Basford}{McLachlan and
  Basford}{1988}]{mclachlan1988mixture}
McLachlan, G.~J. and K.~E. Basford (1988).
\newblock {\em Mixture models: Inference and Applications to Clustering}.
\newblock Marcel Dekker, Inc.

\bibitem[\protect\citeauthoryear{McLachlan, Lee, and Rathnayake}{McLachlan
  et~al.}{2019}]{mclachlan2019finite}
McLachlan, G.~J., S.~X. Lee, and S.~I. Rathnayake (2019).
\newblock Finite mixture models.
\newblock {\em Annual Review of Statistics and its Application\/}~{\em 6},
  355--378.

\bibitem[\protect\citeauthoryear{McLachlan and Peel}{McLachlan and
  Peel}{2004}]{mclachlan2004finite}
McLachlan, G.~J. and D.~Peel (2004).
\newblock {\em Finite Mixture Models}.
\newblock John Wiley \& Sons.

\bibitem[\protect\citeauthoryear{Meinshausen and Rice}{Meinshausen and
  Rice}{2006}]{meinshausen2006estimating}
Meinshausen, N. and J.~Rice (2006).
\newblock Estimating the proportion of false null hypotheses among a large
  number of independently tested hypotheses.
\newblock {\em The Annals of Statistics\/}~{\em 34\/}(1), 373--393.

\bibitem[\protect\citeauthoryear{Nocedal and Wright}{Nocedal and
  Wright}{2006}]{nocedal2006sequential}
Nocedal, J. and S.~J. Wright (2006).
\newblock {\em Numerical Optimization}.
\newblock Springer.

\bibitem[\protect\citeauthoryear{Park, Kim, and Jones}{Park
  et~al.}{2002}]{park2002local}
Park, B., W.~Kim, and M.~Jones (2002).
\newblock On local likelihood density estimation.
\newblock {\em Annals of Statistics\/}~{\em 30\/}(5), 1480--1495.

\bibitem[\protect\citeauthoryear{Patra and Sen}{Patra and Sen}{2016}]{patra}
Patra, R.~K. and B.~Sen (2016).
\newblock Estimation of a two-component mixture model with applications to
  multiple testing.
\newblock {\em Journal of the Royal Statistical Society: Series B\/}~{\em
  78\/}(4), 869--893.

\bibitem[\protect\citeauthoryear{Pu and Arias-Castro}{Pu and
  Arias-Castro}{2020}]{pu2020algorithm}
Pu, X. and E.~Arias-Castro (2020).
\newblock An {EM} algorithm for fitting a mixture model with symmetric
  log-concave densities.
\newblock {\em Communications in Statistics-Theory and Methods\/}~{\em
  49\/}(1), 78--87.

\bibitem[\protect\citeauthoryear{Ridgeway and MacDonald}{Ridgeway and
  MacDonald}{2009}]{ridgeway2009doubly}
Ridgeway, G. and J.~M. MacDonald (2009).
\newblock Doubly robust internal benchmarking and false discovery rates for
  detecting racial bias in police stops.
\newblock {\em Journal of the American Statistical Association\/}~{\em
  104\/}(486), 661--668.

\bibitem[\protect\citeauthoryear{Robin, Reyl{\'e}, Derriere, and Picaud}{Robin
  et~al.}{2003}]{robin2003synthetic}
Robin, A.~C., C.~Reyl{\'e}, S.~Derriere, and S.~Picaud (2003).
\newblock A synthetic view on structure and evolution of the {Milky Way}.
\newblock {\em Astronomy \& Astrophysics\/}~{\em 409\/}(2), 523--540.

\bibitem[\protect\citeauthoryear{Roquain and Verzelen}{Roquain and
  Verzelen}{2020}]{roquain2020}
Roquain, E. and N.~Verzelen (2020).
\newblock False discovery rate control with unknown null distribution: is it
  possible to mimic the oracle?
\newblock {\em arXiv preprint arXiv:1912.03109\/}.

\bibitem[\protect\citeauthoryear{Rudemo}{Rudemo}{1982}]{rudemo1982empirical}
Rudemo, M. (1982).
\newblock Empirical choice of histograms and kernel density estimators.
\newblock {\em Scandinavian Journal of Statistics\/}~{\em 9\/}(2), 65--78.

\bibitem[\protect\citeauthoryear{Samworth}{Samworth}{2018}]{samworth2018recent}
Samworth, R.~J. (2018).
\newblock Recent progress in log-concave density estimation.
\newblock {\em Statistical Science\/}~{\em 33\/}(4), 493--509.

\bibitem[\protect\citeauthoryear{Schuster}{Schuster}{1985}]{schuster1985incorporating}
Schuster, E.~F. (1985).
\newblock Incorporating support constraints into nonparametric estimators of
  densities.
\newblock {\em Communications in Statistics-Theory and Methods\/}~{\em
  14\/}(5), 1123--1136.

\bibitem[\protect\citeauthoryear{Sheather and Jones}{Sheather and
  Jones}{1991}]{sheather1991reliable}
Sheather, S.~J. and M.~C. Jones (1991).
\newblock A reliable data-based bandwidth selection method for kernel density
  estimation.
\newblock {\em Journal of the Royal Statistical Society: Series B\/}~{\em
  53\/}(3), 683--690.

\bibitem[\protect\citeauthoryear{Shen, Levine, and Shang}{Shen
  et~al.}{2018}]{shen2018mm}
Shen, Z., M.~Levine, and Z.~Shang (2018).
\newblock An {MM} algorithm for estimation of a two component semiparametric
  density mixture with a known component.
\newblock {\em Electronic Journal of Statistics\/}~{\em 12\/}(1), 1181--1209.

\bibitem[\protect\citeauthoryear{Silverman}{Silverman}{1986}]{silverman1986density}
Silverman, B.~W. (1986).
\newblock {\em Density Estimation for Statistics and Data Analysis}.
\newblock CRC Press.

\bibitem[\protect\citeauthoryear{Singh, Febbo, Ross, Jackson, Manola, Ladd,
  Tamayo, Renshaw, D'Amico, and Richie}{Singh et~al.}{2002}]{singh2002gene}
Singh, D., P.~G. Febbo, K.~Ross, D.~G. Jackson, J.~Manola, C.~Ladd, P.~Tamayo,
  A.~A. Renshaw, A.~V. D'Amico, and J.~P. Richie (2002).
\newblock Gene expression correlates of clinical prostate cancer behavior.
\newblock {\em Cancer Cell\/}~{\em 1\/}(2), 203--209.

\bibitem[\protect\citeauthoryear{Stone}{Stone}{1984}]{stone1984asymptotically}
Stone, C.~J. (1984).
\newblock An asymptotically optimal window selection rule for kernel density
  estimates.
\newblock {\em The Annals of Statistics\/}~{\em 12\/}(4), 1285--1297.

\bibitem[\protect\citeauthoryear{Storey}{Storey}{2002}]{storey2002direct}
Storey, J.~D. (2002).
\newblock A direct approach to false discovery rates.
\newblock {\em Journal of the Royal Statistical Society: Series B\/}~{\em
  64\/}(3), 479--498.

\bibitem[\protect\citeauthoryear{Tukey}{Tukey}{1960}]{tukey1960survey}
Tukey, J.~W. (1960).
\newblock A survey of sampling from contaminated distributions.
\newblock In I.~Olkin, S.~Ghurye, W.~Hoeffding, W.~Madow, and H.~Mann (Eds.),
  {\em Contributions to Probability and Statistics}, pp.\  448--485. Stanford
  University Press.

\bibitem[\protect\citeauthoryear{van't Wout, Lehrman, Mikheeva, O'Keeffe,
  Katze, Bumgarner, Geiss, and Mullins}{van't Wout
  et~al.}{2003}]{van2003cellular}
van't Wout, A.~B., G.~K. Lehrman, S.~A. Mikheeva, G.~C. O'Keeffe, M.~G. Katze,
  R.~E. Bumgarner, G.~K. Geiss, and J.~I. Mullins (2003).
\newblock Cellular gene expression upon human immunodeficiency virus {T}ype 1
  infection of {CD4}$^+$-{T}-cell lines.
\newblock {\em Journal of Virology\/}~{\em 77\/}(2), 1392--1402.

\bibitem[\protect\citeauthoryear{Walker, Mateo, Olszewski, Gnedin, Wang, Sen,
  and Woodroofe}{Walker et~al.}{2007}]{walker2007velocity}
Walker, M.~G., M.~Mateo, E.~W. Olszewski, O.~Y. Gnedin, X.~Wang, B.~Sen, and
  M.~Woodroofe (2007).
\newblock Velocity dispersion profiles of seven dwarf spheroidal galaxies.
\newblock {\em The Astrophysical Journal Letters\/}~{\em 667\/}(1), L53.

\bibitem[\protect\citeauthoryear{Walther}{Walther}{2009}]{walther2009inference}
Walther, G. (2009).
\newblock Inference and modeling with log-concave distributions.
\newblock {\em Statistical Science\/}~{\em 24\/}(3), 319--327.

\bibitem[\protect\citeauthoryear{Wilson}{Wilson}{1963}]{wilson1963simplicial}
Wilson, R.~B. (1963).
\newblock {\em A Simplicial Algorithm for Concave Programming}.
\newblock Ph.\ D. thesis, Harvard University, Graduate School of Business
  Administration.

\end{thebibliography}

\end{document}